%% file: main.tex
\documentclass[11pt]{article} 

\usepackage[margin=1in]{geometry}
\usepackage{microtype}
\usepackage{graphicx}
\usepackage{subfigure}

\usepackage{algorithm}
\usepackage{algorithmic}
\usepackage{graphicx}
\usepackage{balance} 
\usepackage{mathtools}  
\usepackage{tabulary}
\usepackage{booktabs}

\usepackage{pdfpages}

\usepackage{color, colortbl}
\usepackage{bbm}

\usepackage[
colorlinks=true,
urlcolor=blue,
linkcolor=blue,
citecolor=blue,
]{hyperref}

\usepackage{float}
\usepackage{amssymb,amsfonts,amsmath}
\usepackage{enumerate}
\usepackage{enumitem}
\usepackage{caption}
\usepackage{amsthm}
\usepackage{color}
\usepackage{comment}

\usepackage{mathtools}
\usepackage{amsmath}
\usepackage{amssymb}
 
\usepackage{multirow}
\usepackage{diagbox}

\usepackage{tikz}
\usepackage{pgfplots}
\pgfplotsset{compat=1.5}

\def\AA{\mathbf{A}}
\def\B{\mathbf{B}}

\def\J{\mathbf{J}}
\def\Y{\mathbf{Y}}

\def\U{\mathbf{U}}

\def\V{\mathbf{V}}

\def\S{\mathbf{S}}
\def\T{\mathbf{T}}
\def\W{\mathbf{W}}
\def\Q{\mathbf{Q}}
\def\G{\mathbf{G}}

\def\RR{\mathbf{R}}

\newtheorem{theorem}{Theorem}
\newtheorem{proposition}{Proposition}
\newtheorem{lemma}[theorem]{Lemma}

\newtheorem{observation}[theorem]{Observation}
\newtheorem{definition}[theorem]{Definition}

\newtheorem{corollary}[theorem]{Corollary}

\newcommand{\nnz}{\mathtt{nnz}}
\newcommand{\norm}[1]{\left\lVert#1\right\rVert}
\newcommand{\eps}{\epsilon}
\newcommand{\poly}{\mathtt{poly}}
\newcommand{\R}{\mathbb{R}}
\newcommand{\tr}{\mathtt{tr}}

\usepackage{url}            

\def\Snospace~{\S{}}

\title{In-Database Regression in Input Sparsity Time }

\author{
Rajesh Jayaram\thanks{Rajesh Jayaram and David Woodruff would like to thank the partial support
from the Office of Naval Research (ONR) grant N00014-18-1-2562, and the National Science Foundation (NSF) under
Grant No. CCF-1815840.}\\
Carnegie Mellon University\\
\texttt{rkjayara@cs.cmu.edu} 
\and 
Alireza Samadian \\
University of Pittsburgh \\
\texttt{samadian@cs.pitt.edu} 
\and 
David P. Woodruff\footnotemark[1] \\
Carnegie Mellon University \\
\texttt{dwoodruf@cs.cmu.edu} 
\and
Peng Ye\\
Tsinghua University\\
\texttt{yep17@mails.tsinghua.edu.cn}
}

\begin{document}

\maketitle

\begin{abstract}

Sketching is a powerful dimensionality reduction technique for accelerating algorithms for data analysis. A crucial step in sketching methods is to compute a \textit{subspace embedding} (SE) for a large matrix $\AA \in \R^{N \times d}$. SE's are the primary tool for obtaining extremely efficient solutions for many linear-algebraic tasks, such as least squares regression and low rank approximation.
 Computing an SE often requires an explicit representation of $\AA$ and running time proportional to the size of $\AA$. However, if $\AA= \T_1 \Join \T_2 \Join \dots \Join \T_m$ is the  result of a \textit{database join query} on several smaller tables  $\T_i \in \R^{n_i \times d_i}$, then this running time can be prohibitive, as $\AA$ itself can have as many as $O(n_1 n_2 \cdots n_m)$ rows. 
 
 In this work, we design subspace embeddings for database joins which can be computed significantly faster than computing the join.  For the case of a two table join $\AA = \T_1 \Join \T_2$ we give \textit{input-sparsity} algorithms for computing subspace embeddings, with running time bounded by the number of non-zero entries in $\T_1,\T_2$. This results in input-sparsity time algorithms for high accuracy regression, significantly improving upon the running time of prior FAQ-based methods for regression. We extend our results to arbitrary joins for the ridge regression problem, also considerably improving the running time of prior methods. Empirically, we apply our method to real datasets and show that it is significantly faster than existing algorithms. 
 
\end{abstract}

\input{intro}

\section{Preliminaries}
\label{section:preliminaries}
\input{Preliminaries}

\section{Subspace Embeddings for Two-Table Database Joins}\label{section:twotable}

\input{TwoTable.tex}

\input{Gneral}

\section{Evaluation}\label{section:experiments}

\input{Experiments}

\input{Comparison}

\section{Conclusion}
In this work, we demonstrate that subspace embeddings for database join queries can be computed in time substantially faster than forming the join, yielding input sparsity time algorithms for regression on joins of two tables up to machine precision, and we extend our results to ridge regression on arbitrary joins. Our results improve on the state-of-the-art FAQ-based algorithms for performing in-database regression on joins. Empirically, our algorithms are substantially faster than the state-of-the-art algorithms for this problem.


\bibliography{References}
\bibliographystyle{alpha}

\end{document}

%% file: intro.tex
\section{Introduction}

Sketching is an important tool for dimensionality reduction, whereby one quickly reduces the size of a large-scale optimization problem while approximately preserving the solution space. One can then solve the lower-dimensional problem much more efficiently, and the sketching guarantee ensures that the resulting solution is approximately optimal for the original optimization problem. In this paper, we focus on the notion of a \textit{subspace embedding} (SE), and its applications to problems in databases. Formally, given a large matrix $\AA \in \R^{N \times d}$, an $\eps$-subspace embedding for $\AA$ is a matrix $\S \AA$, where $\S \in \R^{k \times N}$, with the property that 
$\|\S\AA x\|_2 = (1 \pm \eps)\|\AA x\|_2$ 
simultaneously for all $x \in \R^d$. 

A prototypical example of how a subspace embedding can be applied to solve an optimization problem is linear regression, where one wants to solve $\min_x \|\AA x - b\|_2$ for a tall matrix $\AA \in \R^{N \times d}$ where $N \gg d$ contains many data points (rows). Instead of directly solving for $x$, which requires computing the covariance matrix of $\AA$ and which would require $O(Nd^2)$ time for general $\AA$\footnote{This can be sped up to $O(Nd^{\omega-1})$ time in theory, where $\omega \approx 2.373$ is the exponent of matrix multiplication.}, one can first compute a \textit{sketch} $\S \AA,\S b$ of the problem, where $\S \in \R^{k \times N}$ is a random matrix which can be quickly applied to $\AA$ and b. If $[\S \AA,\S b]$ is an $\eps$-subspace embedding for $[\AA,b]$, it follows that the regression problem using the solution $\hat{x}$ to $\min_x \|\S \AA x - \S b\|_2$ will be within a $(1 + \eps)$ factor of the optimal solution cost. However, if $k \ll N$, then solving for $\hat{x}$ can now be accomplished much faster -- in $O(kd^2)$ time. SEs and similar tools for dimensionality reduction can also be used to speed up the running time of algorithms for $\ell_p$ regression, low rank approximation, and many other problems. We refer the reader to the survey \cite{woodruff2014sketching} for a survey of applications of sketching to randomized numerical linear algebra.

One potential downside of most standard subspace embeddings is that the time required to compute $\S \AA$ often scales linearly with the \textit{input sparsity} of $\AA$, meaning the number of non-zero entries of $\AA$, which we denote by $\nnz(\AA)$. This dependence on $\nnz(\AA)$ is in general necessary just to read all of the entries of $\AA$. However, in many applications the dataset $\AA$ is highly structured, and is itself the result of a query performed on a much smaller dataset. A canonical and important example of this is a database join. This example is particularly important since datasets are often
the result of a database join~\cite{hellerstein2012madlib}. In fact, this use-case has motivated companies and teams such as RelationalAI~\cite{RelationalAI} and Google's Bigquery ML~\cite{BigqueryML} to design databases that are capable of handling machine learning queries. Here, we have $m$ tables $\T_1,\T_2,\dots,\T_m$, where $\T_i \in \R^{n_i \times d_i}$, and we can consider their join  $\J=\T_1 \Join \T_2 \Join \dots \Join \T_m \in \R^{N \times d}$ over a set of columns. In general, the number $N$ of rows of $\J$ can be as large as $n_1 n_2 \cdots n_m$, which far exceeds the actual input description of $\sum_i \nnz(\T_i)$ to the problem.

We note that it is possible to do various operations on a join in sublinear time using database algorithms that are developed for Functional Aggregation Queries (FAQ)~\cite{abo2016faq}, and indeed it is possible using so-called FAQ-based techniques \cite{Indatabase:Linear:Regression} to compute the covariance matrix $\J^T \J$ in time $O(d^4 m n \log(n))$ for an acyclic join, where $n = \max(n_1, \allowbreak \ldots, n_m)$, after which one can solve least squares regression in $\text{poly}(d)$ time. While this can be significantly faster than the $\Omega(N)$ time required to compute the actual join, it is still significantly larger than the input description size $\sum_i \nnz(\T_i)$, which even if the tables are dense, is at most $dmn$. When the tables are sparse, e.g., $O(n)$ non-zero entries in each table, one could even hope for a running time close to $O(mn)$. One could hope to achieve such running times with the help of subspace embeddings, by first reducing the data to a low-dimensional representation, and then solving regression exactly on the small representation. However, due to the lack of a clean algebraic structure for database joins, it is not clear how to apply a subspace embedding without first computing the join. Thus, a natural question is: 
\begin{center}
{\it Is it possible to apply a subspace embedding to a join, without having to explicitly form the join?}
\end{center}
We note that the lack of input-sparsity time algorithms for regression on joins is further exacerbated in the presence of categorical features. 
Indeed, it is a common practice to convert categorical data to their so-called {\it one-hot encoding} before optimizing any statistical model. Such an encoding creates one column for each possible value of the categorical feature and only a single column is non-zero for each row. Thus, in the presence of categorical features, the tables in the join are extremely high dimensional and extremely sparse. Since the data is high-dimensional, one often regularizes it to avoid overfitting, and so in addition to ordinary regression, one could also ask to solve regularized regression on such datasets, such as ridge regression. One could then ask if it is possible to design input-sparsity time algorithms on joins for regression or ridge regression.

\subsection{Our Contributions}

We start by describing our results for least squares regression in the important case when the join is on two tables. We note that the two-table case is a very well-studied case, see, e.g., \cite{ams99,agms02,GWWZ15}. Moreover, one can always reduce to the case of a two-table join by precomputing part of a general join. 
The following two theorems state our results for computing a subspace embedding and for solving regression on the join $\J = \T_1 \Join \T_2$ of two tables. Our results demonstrate a substantial improvement over all prior algorithms for this problem. In particular, they answer the two questions above, showing that it is possible to compute a subspace embedding in input-sparsity time, and that regression can also be solved in input-sparsity time. 

To the best of our knowledge, the fastest algorithm for linear regression on two tables has a worst-case time complexity of $\tilde{O} (n d + n D^2)$, where $n = \max(n_1, n_2)$, $d$ is the number of columns, and $D$ is the dimensionality of the data after encoding the categorical data. Note that in the case of numerical data (dense case) $D=d$ since there is no one-hot encoding and the time complexity is $O(n d^2)$, and it can be further improved to $O(n d^{\omega-1})$  where $\omega < 2.373$ is the exponent of fast matrix multiplication; this time complexity is the same as the fastest known time complexity for exact linear regression on a single table. In the case of categorical features (sparse data), using sparse tensors, $D^2$ can be replaced by a constant that is at least the number of non-zero elements in $\J^T \J$ (which is at least $d^2$ and at most $D^2$) using the algorithm in \cite{Indatabase:Linear:Regression}.

We state two results with differing leading terms and low-order additive terms, as one may be more useful than the other depending on whether the input tables are dense or sparse. 


\begin{theorem}[In-Database Subspace Embedding]\label{thm:subspace} 
	 Suppose  $\J=\T_1 \Join \T_2 \in \R^{N \times d}$ is a join of two tables, where $\T_1 \in \R^{n_1 \times d_1},\T_2 \in \R^{n_2 \times d_2}$. Then Algorithm \ref{alg:1} outputs a sketching matrix $\S^* \in \R^{k \times N}$ such that $\tilde{\J} = \S^* \J$ is an $\eps$-subspace embedding for $\J$, meaning
	\[\|\S^*\J x\|_2^2 = (1 \pm \eps)\|\J x\|_2^2 \]
	simultaneously for all $x \in \R^{d}$ with probability\footnote{We remark that using standard techniques for amplifying the success probability of an SE (see Section 2.3 of \cite{woodruff2014sketching}) one can boost the success probability to $1-\delta$ by repeating the entire algorithm $O(\log \delta^{-1})$ times, increasing the running time by a multiplicative $O(\log \delta^{-1})$ factor. One must then compute the SVD of each of the $O(\log \delta^{-1})$ sketches, which results in an additive $O( k d^{\omega-1}\log \delta^{-1} )$ term in the running time, where $\omega \approx 2.373$ is the exponent of matrix multiplication. Note that this additive dependence on $d$ is only slightly ($\approx d^{.373}$) larger than the dependence required for constant probability as stated in the theorem.} at least $9/10$. The running time to return $\S^* \J$ is the minimum of $\tilde{O}((n_1 + n_2)d/\eps^2  +  d^{3}/\eps^2  )$ and  $\tilde{O}((\nnz(\T_1) + \nnz(\T_2))/\eps^2 + (n_1 + n_2)/\eps^2 +d^5/\eps^2   )$.\footnote{We use $\tilde{O}$ notation to omit factors of $\log(N)$.} In the former case, we have $k = \tilde{O}(d^2 /\eps^2)$, and in the latter case we have $k = \tilde{O}(d^4 /\eps^2)$.
\end{theorem}

Next, by following a standard reduction from a subspace embedding to an algorithm for regression, we obtain extremely efficient machine precision regression algorithms for two-table database joins. 

\vspace{.1in}

\begin{theorem}[Machine Precision Regression]\label{thm:reg} 
	\; Suppose  $\J=\T_1 \Join \T_2 \in \R^{N \times d}$ is a join of two tables, where $\T_1 \in \R^{n_1 \times d_1},\T_2 \in \R^{n_2 \times d_2}$. Let $\U \subseteq [d]$ be any subset, and let $\J_U \in \R^{N \times |U|}$ be $\J$ restricted to the columns in $U$, and let $b \in \R^N$ be any column of the join $\J$. Then there is an algorithm which outputs $\hat{x} \in \R^{|U|}$ such that with probability $9/10$\footnote{The probability of success here is the same as the probability of success of constructing a subspace embedding; see earlier footnote about amplifying this success probability.} we have
	\[     \|\J_{U} \hat{x} - b\|_2 \leq (1+\eps) \min_{x \in \R^{|U|}}\|\J_{U} x - b\|_2 .\]
	The running time required to compute $\hat{x}$ is the minimum of $\tilde{O}( ((n_1 + n_2)d  +  d^{3}) \log(1/\eps) )$ and  $\tilde{O}( (\nnz(\T_1) + \nnz(\T_2)  +d^5 ) \log(1/\eps) )$.
\end{theorem}

\paragraph{General Joins}
We next consider arbitrary joins on more than two tables. In this case, we primarily focus on the ridge regression problem $\min_x \|\J x-b\|_2^2 + \lambda \|x\|_2^2,$ for a regularization parameter $\lambda$. This problem is a popular regularized variant of regression and was considered in the context of database joins in \cite{Indatabase:Linear:Regression}. We introduce a general framework to apply sketching methods over arbitrary joins in Section \ref{sec:general}; our method is able to take a sketch with certain properties as a black box, and can be applied both to TensorSketch \cite{avron2014subspace,pagh2013compressed,pham2013fast}, as well as recent improvements to this sketch  \cite{ahle2020oblivious,Woodruff2020near} for certain joins. Unlike previous work, which required computing $\J^T \J$ exactly, we show how to use sketching to approximate this up to high accuracy, where the number of entries of $\J^T \J$ computed depends on the so-called statistical dimension of $\J$, which can be much smaller than the data dimension $D$.


\paragraph{Evaluation}
Empirically, we compare our algorithms on various databases to the previous best FAQ-based algorithm of \cite{Indatabase:Linear:Regression}, which computes each entry of the covariance matrix $\J^T \J$. For two-table joins, we focus on the standard regression problem. We use the algorithm described in Section \ref{section:twotable}, replacing the Fast Tensor-Sketch with Tensor-Sketch for better practical performance. For general joins, we focus on the ridge regression problem; such joins can be very high dimensional and ridge regression helps to prevent overfitting. We apply our sketching approach to the entire join and obtain an approximation to it, where our complexity is in terms of the statistical dimension rather than the actual dimension $D$. Our results demonstrate significant speedups over the previous best algorithm, with only a small sacrifice in accuracy. For example, for the join of two tables in the MovieLens data set, which has 23 features, we obtain a 10-fold speedup over the FAQ-based algorithm, while maintaining a $0.66\%$ relative error. For the natural join of three tables in the real MovieLens data set, which is a join with 24 features, we obtain a 3-fold speedup over the FAQ-based algorithm with only $0.28\%$ MSE relative error. 
Further details can be found in Section \ref{section:experiments}. 

\subsection{Related Work on Sketching Structured Data} 
The use of specialized sketches for different classes of structured matrices $\AA$ has been a topic of substantial interest. 
The TensorSketch algorithm of \cite{pagh2013compressed} can be applied to Kronecker products $\AA = \AA_1 \otimes \cdots \otimes \AA_m$ without explicitly computing the product. The efficiency of this algorithm was recently improved by \cite{ahle2020oblivious}.
The special case when all $\AA_i$ are equal is known as the polynomial kernel, which was considered in \cite{pham2013fast} and extended 
by \cite{avron2014subspace}. 

Kronecker Product Regression has also been studied for the $\ell_p$ loss functions \cite{diao2018,diao2019optimal}, which also gave improved algorithms for $\ell_2$ regression. 
In \cite{asw13,ShiW19} 
it is shown that regression on $\AA$ can be solved in time $T(\AA)\cdot \poly(\log(nd))$, where $T(\AA) \leq \nnz(\AA)$ is the time needed to compute the matrix-vector product $\AA y$ for any $y \in \R^d$. For many classes of $\AA$, such as Vandermonde matrices, $T(\AA)$ is substantially smaller than $\nnz(\AA)$. 

Finally, a flurry of work has used sketching to obtain faster algorithms for low rank approximation of structured matrices. In \cite{musco2017sublinear,Bakshi2019RobustAS}, low rank approximations to positive semidefinite (PSD) matrices are computed in time sublinear in the number of entries of $\AA$. This was also shown for distance matrices in \cite{bakshi2018sublinear, indyk2019sample,Bakshi2019RobustAS}.

\subsection{Related In-Database Machine Learning Work}
The work of \cite{abo2016faq} introduced Inside-out, a polynomial time algorithm for calculating functional aggregation queries (FAQs) over joins without performing the joins themselves, which can be utilized to train various types of machine learning models. 
The Inside-Out algorithm builds on several earlier papers, including~\cite{aji2000generalized,Dechter:1996,Kohlas:2008,GM06}. Relational linear regression, singular value decomposition, and factorization machines are studied extensively in multiple prior works~\cite{rendle2013scaling,Kumar:2015:LGL:2723372.2723713,SystemF,khamis2018ac,Indatabase:Linear:Regression,Kumar:2016:JJT:2882903.2882952,elgamal2017spoof,kumar2015demonstration}. The best known time complexity for training linear regression when the features are continuous, is $O(d^4 m n^{\text{fhtw}} \log(n))$ where fhtw is the fractional hypertree width of the join query.  Note that the fractional hypertree width is $1$ for acyclic joins.  For categorical features, the time complexity is $O(d^2 m n^{\text{fhtw}+2})$ in the worst-case; however, \cite{Indatabase:Linear:Regression} uses sparse computation of the results to reduce this time depending on the join instance. In the case of polynomial regression, the calculation of pairwise interactions among the features can be time-consuming and it is addressed in~\cite{Indatabase:Linear:Regression,li2019enabling}. 
 A similar line of work \cite{Arenas19, 10.1145/3406325.3465353, Arenas2020when,focke2021approximately} has developed polynomial time algorithms for sampling from and estimating the size of certain database queries, including join queries with bounded fractional hypertree width, without fully computing the joins themselves. 

Relational support vector machines with Gaussian kernels are studied in \cite{yangtowards}. In \cite{cheng2019nonlinear}, a relational algorithm is introduced for Independent Gaussian Mixture Models, which can be used for kernel density estimation by estimating the underlying distribution of the data.

%% file: Preliminaries.tex
\subsection{Database Joins}\label{sec:databasenotation}

We first introduce the notion of a \textit{block} of a join, which will be important in our analysis. 
Let $\T_1,\dots,\T_m$ be tables, with $\T_i \in \R^{n_i \times d_i}$. 
 Let $\J=\T_1 \Join \T_2 \Join \dots \Join \T_m \in \R^{N \times d}$ be an arbitrary join on the tables $\T_i$.
 Let $Q$ be the subset of columns which are contained in at least two tables, e.g., the columns which are joined upon.  For any subset $U$ of columns and any table $T$ containing a set of columns $U'$, let $T|_U$ be the projection of $T$ onto the columns in $U \cap U'$.  Similarly define $r|_U$ for a row $r$. Let $C$ be the set of columns in $\J$, and let $C_j \subset C$ be the columns contained in $T_j$.
Define the set of \textit{blocks} $\mathcal{B} = \mathcal{B}(\J) \subset \R^{|Q|}$ of the join to be the set of distinct rows in the projection of $\J$ onto $Q$. In other words, $\mathcal{B}$ is the set of distinct rows which occur in the restriction of $\J$ to the columns being joined on.
 For any $j \in [m]$, let $\hat{\T}_j \in \R^{n_j \times d}$ be the embedding of the rows of $\T_j$ into the join $\J$, obtained by padding $\T_j$ with zero-valued columns for each column not contained in $\T_j$, and such that for any column $c$ contained in more than one $\T_j$, we define the matrices $\hat{\T}_j$ so that exactly one of the $\hat{\T}_j$ contains $c$ (it does not matter which of the tables containing the column $c$ has $c$ assigned to it in the definition of $\hat{\T}_j$). More formally, we fix any partition $\{\hat{C}_j\}_{j \in [m]}$ of $C$, such that $C = \cup_j \hat{C}_j$ and $\hat{C}_j \subseteq C_j$ for all $j$.

 For simplicity, given a block $\vec{i} \in \mathcal{B}$, which was defined as a row in $\R^{|Q|}$, we drop the vector notation and write $\vec{i} = i \in \mathcal{B}$. 
 For a given $i =  (i_1,\dots,i_{|Q|})  \in \mathcal{B}$, let $s_{(i)}$ denote the \textit{size} of the block, meaning the number of rows $r$ of the join $\J$ such that $i_j$ is in the $j$-th column of $r$ for all $j \in Q$. For $i \in \mathcal{B}$, let $\T_j^{(i)}$ be the subset of rows $r$ in $\T_j$ such that $r|_Q =i|_{C_j}$, and similarly define $\hat{\T}_j^{(i)}, \J^{(i)}$ to be the subset of rows $r$ in $\hat{\T}_j$ (respectively $\J$) such that $r|_Q = i|_{\hat{C}_j}$ (respectively $r|_Q = i$). 
 For a row $r$ such that $r|_Q = i$ we say that $r$ ``belongs'' to the block $i \in \mathcal{B}$.
 Let $s_{(i),j}$ denote the number of rows of $\T_j^{(i)}$, so that $s_{(i)} = \prod_{j=1}^m s_{(i),j}$. 

As an example, considering the join $T_1(A,B) \Join T_2(B,C)$, we have one block for each distinct value of $B$ that is present in both $T_1$ and $T_2$, and for a given block $B=b$, the size of the block can be computed as the number of rows in $T_1$, with $B=b$, multiplied by the number of rows in $T_2$, with $B=b$.
 
 Using the above notion of blocks of a join, we can construct $\J$ as a stacking of matrices $\J^{(i)}$ for $i \in \mathcal{B}$. For the case of two table joins $\J = \T_1 \Join \T_2$, we have $\J^{(i)} = \left(\hat{\T}_1^{(i)} \otimes \mathbf{1}^{s_{(i),2}} +  \mathbf{1}^{s_{(i),1}} \otimes   \hat{\T}_2^{(i)}\right) \in \R^{s_{(i)} \times d}$. In other words, $\J^{(i)}$ is the subset of rows of $\J$ contained in block $i$. Observe that the entire join $\J$ is the result of stacking the matrices $\J^{(i)}$ on top of each other, for all $i \in \mathcal{B}$. In other words, if $\mathcal{B} = \{i_1,i_2,\dots, i_{|\mathcal{B}|}\}$, the join $\J = \T_1 \Join \T_2$ is given by 
$ \J = \left[
 (\J^{(i_1)})^T ,
 (\J^{(i_2)})^T,
 \dots ,
 (\J^{(i_{|\mathcal{B}|})})^T \right]^T$.
 
Figure \ref{fig:example_two_tables} illustrates an example of blocks in a two table join. In this example column $f_2$ is the column that we are joining the two tables on, and there are two values for $f_2$ that are present in both tables, namely the values$\{1,2\}$. Thus $\mathcal{B} = \{B_1,B_2\}$, where $B_1 =1$ and $B_2 = 2$. In other words, Block $B_1$ is the block for value $1$, and its size is $s_1 = 4$, and similarly $B_2$ has size $s_2 = 4$. Figure \ref{fig:example_two_tables} illustrates how the join $\mathcal{J}$ can be written as stacking together the block-matrices $\J^{(1)}$ and $\J^{(2)}$. Figure \ref{fig:example:partial_tables} shows the tables $T_i^{(j)}$ for different values of $i$ and $j$ in the same example.

\begin{figure}[t!]
     \centering
     \includegraphics[scale=0.65]{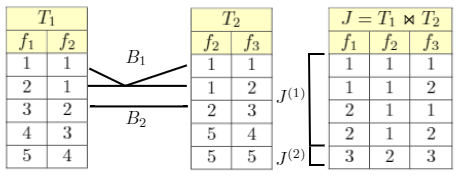}
     \caption{Example of two table join blocks}
     \label{fig:example_two_tables}
 \end{figure}
 
\begin{figure}[t!]
\centering
\begin{tabular}{|c|c|}
\hline
\rowcolor[HTML]{FFFFC7} 
\multicolumn{2}{|c|}{\cellcolor[HTML]{FFFFC7}$T_1^{(1)}$} \\ \hline
\rowcolor[HTML]{FFFFC7} 
$f_1$                    & $f_2$                    \\ \hline
1                        & 1                        \\ \hline
2                        & 1                        \\ \hline
\end{tabular}
\quad
\begin{tabular}{|c|c|}
\hline
\rowcolor[HTML]{FFFFC7} 
\multicolumn{2}{|c|}{\cellcolor[HTML]{FFFFC7}$T_1^{(2)}$} \\ \hline
\rowcolor[HTML]{FFFFC7} 
$f_1$                    & $f_2$                    \\ \hline
3                        & 2                        \\ \hline
\end{tabular}
\quad
\begin{tabular}{|c|c|}
\hline
\rowcolor[HTML]{FFFFC7} 
\multicolumn{2}{|c|}{\cellcolor[HTML]{FFFFC7}$T_2^{(1)}$} \\ \hline
\rowcolor[HTML]{FFFFC7} 
$f_2$                    & $f_3$                    \\ \hline
1                        & 1                        \\ \hline
1                        & 2                        \\ \hline
\end{tabular}
\quad
\begin{tabular}{|c|c|}
\hline
\rowcolor[HTML]{FFFFC7} 
\multicolumn{2}{|c|}{\cellcolor[HTML]{FFFFC7}$T_2^{(2)}$} \\ \hline
\rowcolor[HTML]{FFFFC7} 
$f_2$                    & $f_3$                    \\ \hline
2                        & 3                        \\ \hline
\end{tabular}
\caption{Examples of $T_i^{(j)}$}
\label{fig:example:partial_tables}
\end{figure}

Finally, for any subset $U \subseteq [N]$, let $\J_{U}$ denote the set of rows of $\J$ belonging to $U$. If $L$ is a set of blocks of $J$, meaning $L \subseteq \mathcal{B}(\J)$, then let $\J_L$ denote the set of rows of $\J$ belonging to some block $i \in L$ (recall that a row $r$ ``belongs'' to a block $i \in L \subseteq \mathcal{B}$ if $r|_Q = i$). A table of notation summarizing the above can be found in Figure \ref{tab:TableOfNotation}. 

\begin{table}[h]\caption{Table of Notation}
\begin{center}
\begin{tabular}{r c p{6cm} }
\toprule
$T_{i}$ & $\triangleq$ & a $n_i \times d_i$ sized table\\
$\J$ & $\triangleq$ & join of all tables, i.e., $\J = T_1 \Join T_2 \Join \dots \Join T_m$ \\
$C$ & $\triangleq$ & set of columns of $\J$\\
$C_i$ & $\triangleq$ & set of columns of $T_i$\\
$\hat{C}_i$ & $\triangleq$ & Partition of $C$ such that $\hat{C}_i \subseteq C_i$ for $i \in [m]$. \\
$T|_U$ & $\triangleq$ & projection of $T$ onto columns $U \cap U'$, where $U'$ are the columns of $T$\\
$\hat{T}_{i}$ & $\triangleq$ &  Result of padding $T_i|_{\hat{C}_i}$ with zero-valued columns in $C \setminus \hat{C}_i$\\
$\mathcal{B}$ & $\triangleq$ & set of \textit{blocks}, i.e., distinct rows of $\J|_Q$\\
$s_{(i)}$ & $\triangleq$ & size of block $i$, i.e., number of rows $r$ of $\J$ with $r|_q = i \in \mathcal{B}$\\  
$T_j^{(i)}$ & $\triangleq$ & subset of rows $r$ in $\T_j$ such that $r|_Q =i|_{C_j}$ \\
$\hat{T}_j^{(i)}$ & $\triangleq$ & subset of rows $r$ in $\hat{T}_j$ such that $r|_Q =i|_{\hat{C}_j}$ \\
$\J_j^{(i)}$ & $\triangleq$ & subset of rows $r$ in $\J$ such that $r|_Q =i$ \\
\bottomrule
\end{tabular}
\end{center}
\label{tab:TableOfNotation}
\end{table}

\subsection{Background for General Database Joins}\label{appendix:database:background}

We begin with some additional definitions relating to database joins.
\begin{definition}[Join Hypergraph]
Given a join $\J=\T_1 \Join \dots \Join \T_m$, the hypergraph associated with the join is $H=(V,E)$ where $V$ is the set of vertices and for every column $c_i$ in $J$, there is a vertex $v_i$ in $V$, and for every table $\T_i$ there is a hyper-edge $e_i$ in $E$ that has the vertices associated with the columns of $\T_i$.
\end{definition}

\begin{definition}[Acyclic Join]
We call a join query \textbf{acyclic} if one can repeatedly apply one of the two operations and convert the query to an empty query:
\begin{enumerate}
    \item remove a column that is only in one table.
    \item remove a table for which its columns are fully contained in another table.
\end{enumerate}
\end{definition}

\begin{definition}[Hypergraph Tree Decomposition]
Let $H=(V,E)$ be a hypergraph and $T=(V',E')$ be a tree on a set of vertices, where each vertex $v' \in V'$ is called the \textbf{bag} of $v'$, denoted by $b(v')$, and corresponds to a subset of vertices of $V$. Then $T$ is called a \textbf{hypergraph tree decomposition} of $H$ if the following holds:
\begin{enumerate}
    \item for each hyperedge $e \in E$, there exists $v' \in V'$ such that $e \subseteq b(v')$, and 
    \item for each vertex $v \in V$, the set of vertices in $V'$ that have $v$ in their bag is non-empty and they form a connected subtree of $T$.
\end{enumerate}
\end{definition}

\begin{definition}
Let $H=(V,E)$ be a join hypergraph and $T=(V',E')$ be its tree decomposition. For each $v' \in V'$, let $X^{v'} = \allowbreak (x_1^{v'},x_2^{v'}, \allowbreak \dots, x_m^{v'})$ be the optimal solution to the following linear program: $\texttt{min}  \sum_{j=1}^{t} x_{j}$,  $\text{subject to }  \sum_{j:v_{i} \in e_{j} }x_{j} \geq 1, \forall v_i \in b(v')$ where $0 \leq x_{j} \leq 1$ for each $j \in [t]$.
Then the \textbf{width of $v'$} is $\sum_i x^{v'}_i$, denoted by $w(v')$, and the \textbf{fractional width of $T$} is $\max_{v' \in V'} w(v')$.
\end{definition}

\begin{definition}[fhtw]
Given a join hypergraph $H=(V,E)$, the \textbf{fractional hypertree width of $H$}, denoted by fhtw, is the minimum fractional width of its hypergraph tree decomposition. Here the minimum is taken over all possible hypertree decompositions.
\end{definition}

\begin{observation}
The fractional hypertree width of an \allowbreak acyclic join is $1$, and each bag in its hypergraph tree decomposition is a subset of the columns in some input table.
\end{observation}

\begin{definition}[FAQ]
Let $J=T_1\Join\dots \Join T_m$ be a join of $m$ input tables. For each table $T_i$, let $F_i: T_i \to S$ be a function mapping the rows of $T_i$ to a set $S$. For every row $X\in J$, let $X_i$ be the projection of $X$ onto the columns of $T_i$. Then the following is a SumProd Functional Aggregation Query (FAQ):
\begin{align}
    \bigoplus_{X \in J} \bigotimes_i F_i(X_i)
\end{align}
where $(S,\oplus,\otimes)$ is a commutative semiring.
\end{definition}

\begin{theorem}[\cite{abo2016faq}]
\label{thm:inside_out}
Inside-out is an algorithm which computes the result of a FAQ in time $O(T md^2n^{\text{fhtw}}\log(n))$ where $m$ is the number of tables, $d$ is the number of columns in $J$, $n$ is the maximum number of rows in any input table, $T$ is the time to compute the operators $\oplus$ and $\otimes$ on a pair of operands, and $fhtw$ is the fractional hypertree width of the query.
\end{theorem}

In \cite{Indatabase:Linear:Regression}, given a join $J=T_1 \Join \dots \Join T_m$, it is shown that the entries of $J^T J$ can be expressed as a FAQ and computed using the inside-out algorithm.

\subsection{Linear Algebra}
We use boldface font, e.g., $\AA,\B,\J$, throughout to denote matrices.
Given $\AA \in \R^{n \times d}$ with rank $r$, we write $\AA = \U \mathbf{\Sigma} \V^T$ to denote the singular value decomposition (SVD) of $\AA$, where $\U \in \R^{n \times r}, \V \in \R^{d \times r}$, and $\mathbf{\Sigma} \in \R^{ r \times r}$ is a diagonal matrix containing the non-zero singular values of $\AA$. For $i \in [d]$, we write $\sigma_i(\AA)$ to denote the $i$-th (possibly zero-valued) singular value of $\AA$, so that $\sigma_1(\AA) \geq \sigma_2(\AA) \geq \dots \geq \sigma_d(\AA)$. We also use $\sigma_{\max(\AA)}$ and $\sigma_{\min}(\AA)$ to denote the maximum and minimum singular values of $\AA$ respectively, and let $\kappa(\AA) = \frac{\sigma_{\max}(\AA)}{\sigma_{\min}(\AA)}$ denote the condition number of $\AA$. Let $\AA^+$ denote the Moore-Penrose pseudoinverse of $\AA$, namely $\AA^+ =  \V \mathbf{\Sigma}^{-1} \U^T$. Let $\|\AA\|_F = (\sum_{i,j} \AA_{i,j}^2 )^{1/2}$ denote the Frobenius norm of $\AA$, and $\|\AA\|_2 = \sigma_{\max}(\AA)$ the spectral norm. We write $\mathbb{I}_n \in \R^{n \times n}$ to denote the $n$-dimensional identity matrix. For a matrix $\AA \in \R^{n \times d}$, we write $\nnz(\AA)$ to denote the number of non-zero entries of $\AA$. We can assume that each row of the table $\T_j$ is non-zero, since otherwise the row can be removed, and thus $\nnz(\T_j) \geq n_i$. Given $\AA \in \R^{n \times d}$, we write $\AA_{i,*} \in \R^{1 \times d}$ to denote the $i$-th row (vector) of $\AA$, and $\AA_{*,i} \in \R^{n \times 1}$ to denote the $i$-th column (vector) of $\AA$.

For values $a,b \in \R$ and $\eps >0$, we write $a = (1 \pm \eps)b$ to denote the containment $(1-\eps)b \leq a \leq (1+\eps) b$. For $n \in \mathbb{Z}_+$, let $[n] = \{1,2,\dots,n\}$. Throughout, we will use  $\tilde{O}(\cdot)$ notation to omit poly$(\log N)$ factors. 

\begin{definition}[Statistical Dimension]
For a matrix $\AA \in \R^{n \times d}$, and a non-negative scalar $\lambda$, the $\lambda$-statistical dimension is defined to be $d_\lambda = \sum_i \frac{\lambda_i(\AA^T \AA)}{\lambda_i(\AA^T \AA) + \lambda}$, where $\lambda_i(\AA^T \AA)$ is the $i{\text{-th}}$ eigenvalue of $\AA^T \AA$.
\end{definition}

\begin{definition}[Subspace Embedding]
For an $\eps \geq 0$, we say that $\tilde{\AA} \in \R^{m \times d}$ is an $\eps$-subspace embedding for $\AA \in \R^{n \times d}$ if for all $x \in \R^d$ we have
\[  (1 - \eps)\|\AA x\|_2 \leq \|\tilde{\AA} x\|_2 \leq (1+\eps) \|\AA x\|_2. \]    
\end{definition}
Note that if $\tilde{\AA} \in \R^{m \times d}$ is an $\eps$-subspace embedding for $\AA \in \R^{n \times d}$, in particular this implies that $\sigma_i(\AA) = (1 \pm \eps) \sigma_i(\tilde{\AA})$ for all $i \in [d]$.

\subsubsection{Leverage Scores} \label{app:leverage}
 The leverage score of the $i$-th row $\AA_{i,*}$ of $\AA \in \R^{n \times d}$ is defined to be $\tau_i(\AA) =     \AA_{i,*} (\AA^T \AA)^+ \AA_{i,*}^T$. 
Let $\tau(\AA) \in \R^{n}$ be the vector such that $(\tau(\AA))_i = \tau_i(\AA)$. Then $\tau(\AA)$ is the diagonal of $\AA  (\AA^T \AA)^+ \AA^T$, which is a projection matrix. Thus $\tau_i(\AA) \leq 1$ for all $i \in [n]$. It is also easy to see that $\sum_{i=1}^n \tau_i(\AA) \leq d$ \cite{cohen2015uniform}. Our algorithm will utilize the \textit{generalized leverage scores}. Given matrices $\AA \in \R^{n \times d}$ and $\B \in \R^{n_1 \times d}$, the generalized leverage scores of $\AA$ with respect to $\B$ are defined as 
\[ \tau_i^{\B}( \AA)  = \begin{cases}
\AA_{i,*} (\B^T \B)^+ \AA_{i,*}^T  & \text{if } \AA_{i,*} \bot \text{ker}(\B) \\
1 & \text{otherwise}  \\
\end{cases}\] 
We remark that in the case were $\AA_{i,*}$ has a component in the kernel (null space) of $B$, denoted by $\text{ker}(\B)$, $\tau_i^{\B}( \AA)$ is defined to be $\infty$ in \cite{cohen2015uniform}. However, as stated in that paper, this definition was simply for notational convenience, and the results would equivalently hold setting $\tau_i^{\B}( \AA) =1$ in this case.  
 Note that for a matrix $\AA \in \R^{n \times d}$ with SVD $\B = \U \mathbf{\Sigma} \V^T$, we have $(\B^T \B)^+ = \V \Sigma^{-2} \V^T$. Thus, for any $x \in \R^d$, we have $x^T (\B^T \B)^+ x = \|x^T \V \Sigma^{-1}\|_2^2$, and in particular $\tau_i^{\B}( \AA) = \|\AA_{i,*}^T \V \Sigma^{-1}\|_2^2$ if  $\AA_{i,*} \bot \text{ker}(\B)$, where $\AA_{i,*} \bot \text{ker}(\B)$ means $\AA_{i,*}$ is perpendicular to the kernel of $\B$.
 
\begin{proposition}\label{prop:tausubspace}
If $\B' \in \R^{n_1' \times d}$ is an $\eps$-subspace embedding for $\B \in \R^{n_1 \times d}$, and $\AA \in \R^{n \times d}$ is any matrix, then $\tau_i^{\B'}(\AA) = (1 \pm O(\eps)) \tau_i^{\B}(\AA)$
\end{proposition}
\begin{proof}
If $\B' $ is an $\eps$-subspace embedding for $\B$, then the spectrum of $(\B' (\B')^T)^+$ is a $(1 \pm \eps)^{-2}$ approximation to the spectrum of  $(\B \B^T)^+$, so $x^T(\B' (\B')^T)^+x = (1 \pm \eps)^{-2} x^T(\B \B^T)^+x$ for all $x \in \R^d$, which completes the proof. 
\end{proof}

Our algorithm will employ a mixture of several known oblivious subspace embeddings as tools to construct our overall database join SE. The first result we will need is an improved variant of \textit{Tensor-Sketch}, which is an SE that can be applied quickly to tensor products of matrices. 


\begin{lemma}[Fast Tensor-Sketch,  Theorem \allowbreak 3  of \cite{ahle2020oblivious}]\label{lem:fastTenSketch}
	Fix any matrices $\AA_1, \AA_2, \dots, \AA_m$, where $\AA_i \in \R^{n_i \times d_i} $, fix $\eps >0$ and $\lambda \geq 0$. Let $n = n_1 \cdots n_m$ and $d = d_1 \cdots d_m$. Let $\AA = \AA_1 \otimes \AA_2 \otimes \cdots \otimes \AA_m$ have statistical dimension $d_\lambda$. Then there is an oblivious randomized sketching algorithm which produces a matrix  $\S \in \R^{k \times n}$, where $k = O(d_\lambda m^4 / \eps^2)$, such that with probability $1-1/\poly(n)$, we have that for all $x \in \R^d$
	\[ \|\S \AA x\|_2^2 + \lambda \|x\|_2^2 = (1 \pm \eps)(\|\AA x\|_2^2  +  \lambda \|x\|_2^2).\]
	Note for the case of $\lambda = 0$, this implies that  $\S\AA$ is an $\eps$-subspace embedding for $\AA$. Moreover, $\S\AA$ can be computed in time $\tilde{O}( \sum_{i=1}^m \nnz(\AA_i)/\eps^2\cdot m^5 + kdm)$.\footnote{Theorem 3 of \cite{ahle2020oblivious} is written to be applied to the special case of the polynomial kernel, where $A_1 = A_2 = \dots = A_m$. However, the algorithm itself does not use this fact, nor does it require the factors in the tensor product to be non-distinct.} 
\end{lemma}
For the special case of $\lambda=0$ in Lemma \ref{lem:fastTenSketch}, the statistical dimension is $d$, and Tensor-Sketch is just a standard SE. 
\begin{lemma}[OSNAP Transform \cite{nelson2013osnap}]\label{lem:OSNAP}
Given any $\AA \allowbreak \in \R^{N \times d}$, there is a randomized oblivious sketching algorithm that produces a matrix $\W \in \R^{t \times N}$ with $t = \tilde{O}(d/\eps^2)$, such that $\W\AA$ can be computed in time $\tilde{O}(\nnz(\AA)/\eps^2)$, and such that $\W \AA$ is an $\eps$-subspace embedding for $\AA$ with probability at least $99/100$. Moreover, each column of $\W$ has at most $\tilde{O}(1)$ non-zero entries.
\end{lemma}


\begin{lemma}[Count-Sketch\cite{10.1145/2488608.2488620}]\label{lem:CountSketch}
  For any fixed matrix $\AA\in\mathbb{R}^{n\times d}$, and any $\eps>0$, there exists an algorithm which produces a matrix $\S\in\mathbb{R}^{k\times n}$, where $k = O(d^2/\eps^2)$, such that $\S\AA$ is an $\eps$-subspace embedding for $\AA$ with probability at least $99/100$. Moreover, each column of $\S$ contains exactly one non-zero entry, and therefore $\S\AA$ can be computed in $O(\nnz(\AA))$ time.
\end{lemma}

%% file: TwoTable.tex
 In this section, we will describe our algorithms for fast computation of in-database subspace embeddings for joins of two tables $\J = \T_1 \Join \T_2$, where $\T_i \in \R^{n_i \times d_i}$, and $\J \in \R^{N \times d}$. As a consequence of our subspace embeddings, we obtain an input sparsity time algorithm for machine precision in-database regression. Here, machine precision refers to a convergence rate of $\log(1/\eps)$ to the optimal solution.  
 
Our subspace embedding algorithm can be run with two separate hyper-parameterizations, one of which we refer to as the \textit{dense case} where the tables $\T_1,\T_2$ have many non-zero entries, and the other is referred to as the \textit{sparse case}, where we exploit the sparsity of the tables $\T_1,\T_2$. In the former, we will obtain $\tilde{O}(\frac{1}{\eps^2} ((n_1+ n_2)d + d^3))$ runtime for construction of our subspace embedding, and in the latter we will obtain $\tilde{O}( \frac{1}{\eps^2} (\nnz(\T_1) + \nnz(\T_2) + d^5))$ time.  Thus, when the matrices $\T_1,\T_2$ are dense, we have $(n_1+ n_2)d = \Theta(\nnz(\T_1) + \nnz(\T_2))$, in which case the former algorithm has a strictly better runtime dependence on $n_1,n_2$, and $d$. However, for the many applications where $\T_1,\T_2$ are sparse, the latter algorithm will yield substantial improvements in runtime. By first reading off the sparsity of $\T_1,\T_2$ and choosing the hyperparameters which minimize the runtime, the final runtime of the algorithm is the minimum of the two aforementioned runtimes. 
  
 \begin{algorithm}[!ht]
	\caption{Subspace embedding for join  $\J=\T_1 \Join \T_2$. } \label{alg:1}
	
		\begin{algorithmic}[1]
\STATE
		In the \textit{dense} case set $\gamma =1$. In the \textit{sparse} case, case $\gamma = d$. 
Compute block sizes $s_{(i)}$, and let $\mathcal{B}_{\text{big}} = \{i\in \mathcal{B}\; | \; \max_j \{s_{(i),j} \} \geq d \cdot \gamma \}$. Set $\mathcal{B}_{\text{small}} = \mathcal{B}\setminus \mathcal{B}_{\text{big}}$, $n_{\text{small}} = |\mathcal{B}_{\text{small}}|$.	
		
		\STATE	For each $i \in \mathcal{B}_{\text{big}}$, generate a Fast Tensor-Sketch  matrix $\S^i \in \R^{t \times s_{(i)}}$ (Lemma \ref{lem:fastTenSketch}) and compute $\S^{i} \J^{(i)}$.
	\STATE	Let $\tilde{\J}_{\text{big}}$ be the matrix from stacking the matrices $\{\S^{i}\J^{(i)}\}_{i\in\mathcal{B}_{\text{big}}}$. Generate a Count-Sketch matrix $\S'$ (Lemma \ref{lem:CountSketch}) and compute $\S' \tilde{\J}_{\text{big}}$.
		\STATE	Let $\J_{\text{small}} = \J_{\mathcal{B}_{\text{small}}}$, and sample uniformly a subset $U$ of $m = \Theta((n_1 + n_2)/\gamma)$ rows of $\J_{\mathcal{B}_{small}} \in \R^{n_{\text{small} \times d}}$ and form the matrix $\tilde{\J}_{\text{small}} =( \J_{\text{small}})_U \in \R^{m \times d }$. \label{line:3}
					\STATE	Generate OSNAP transform $\W$ (Lemma \ref{lem:OSNAP}) and compute $\W \tilde{\J}_{\text{small}}$ and the SVD $\W\cdot\tilde{\J}_{\text{small}} = \U \Sigma \V^T$. \label{line:4}

		\STATE	Generate Gaussian matrix $\G \in \R^{d \times t}$ with entries drawn i.i.d. from $\mathcal{N}(0,1/t^2)$, $t = \Theta(\log N)$, and Gaussian vector $g\sim  \mathcal{N}(0,\mathbb{I}_d)$. \label{line:5}
		
		\STATE	For all rows $i$ of $\J_{\text{small}}$, set $ \|(\J_{\text{small}})_{i,*} \left(\mathbb{I}_{d} - \V \V^T \right) g \|_2^2 = \alpha_i$, and 
		\[\tilde{\tau}_i = \begin{cases}
		1 & \text{if } \alpha_i > 0 \\
		 \|(\J_{\text{small}})_{i,*} \V \Sigma^{-1} \G\|_2^2  & \text{otherwise} 
		\end{cases}.\] 
	\STATE	Using Algorithms \ref{alg:L2presample} and \ref{alg:L2sample}, construct diagonal row sampling matrix $\S  \in \R^{n_{\text{small}} \times n_{\text{small}}}$   such that $\S_{i,i} = \frac{1}{\sqrt{p_i}}$ with probability $p_i$, and $\S_{i,i}  = 0$ otherwise, where 
	$p_i \geq   \min\left\{1, \frac{\log d }{\eps^2}\cdot\tilde{\tau}_i \right\}$ 		\label{line:6}

	\STATE	Return $\tilde{\J}$, where $\tilde{\J}$ is the result of stacking the matrices $\S'\tilde{\J}_{\text{big}}$ with  $\S \J_{\text{small}}$. 
	 
		\end{algorithmic}
\end{algorithm}

 Our main subspace embedding is given in Algorithm \ref{alg:1}. We begin by informally describing the algorithm and analysis, before proceeding to formal proofs in Section \ref{sec:twotableanalysis}.
 As noted in Section \ref{sec:databasenotation}, we can describe the join $\J$ as the result of stacking several ``blocks'' $\J^{(i)}$, where the rows of $\J^{(i)}$ consist of all pairs of concatenations of a row of $\T_1^{(i)}$ and $\T_2^{(i)}$, where the $\T_j^{(i)}$'s partition $\T_j$. We deal separately with blocks $i$ for which $\J^{(i)}$ contains a very large number of rows, and smaller blocks. Formally, we split the set of blocks $\mathcal{B}(\J)$ into $\mathcal{B}_{\text{big}}$ and $\mathcal{B}_{\text{small}}$. For each block $\J^{(i)}$ from $\mathcal{B}_{\text{big}}$, we apply a fast tensor sketch transform to obtain a subspace embedding for that block.
 
 For the smaller blocks, however, we need a much more involved routine. Our algorithm computes a random sample of the rows of the blocks $\J^{(i)}$ from $\mathcal{B}_{\text{small}}$, denoted $\tilde{\J}_{\text{small}}$. Using the results of \cite{cohen2015uniform}, it follows that sampling sufficiently many rows from the distribution induced by the generalized leverage scores $\tau_i^{\tilde{\J}_{\text{small}}}(\J_{\text{small}})$ of $\J_{\text{small}}$ with respect to $\tilde{\J}_{\text{small}}$ yields a subspace embedding of $\J_{\text{small}}$. However, it is not possible to write down (let alone compute) all the values $\tau_i^{\tilde{\J}_{\text{small}}}(\J_{\text{small}})$, since there can be more rows $i$ in $\J_{\text{small}}$ than our entire allowable running time. 
 
 To handle this issue, we first note that by Proposition \ref{prop:tausubspace} and the discussion prior to it, the value $\tau_i^{\tilde{\J}_{\text{small}}}(\J_{\text{small}})$ is well-approximated by $\|(\J_{\text{small}})_{i,*} \V \Sigma^{-1}\|_2^2$, which in turn is well-approximated by $\|(\J_{\text{small}})_{i,*} \V \Sigma^{-1} \G\|_2^2$ if $\G$ is a Gaussian matrix with only a small $\Theta(\log N)$ number of columns. Thus, sampling from the generalized leverage scores $\tau_i^{\tilde{\J}_{\text{small}}}(\J_{\text{small}})$ can be approximately reduced to the problem of sampling a row $i$ from $\J_{\text{small}} \Y$ with probability proportional to  $\|(\J_{\text{small}})_{i,*} \Y\|_2^2$, where $\Y$ is any matrix given as input. We then design a fast algorithm which accomplishes precisely this task: namely, for any join $\J'$ and input matrix $\Y$ with a small number of columns, it samples rows from $\J'\Y$ with probability proportional to the squared row norms of $\J'\Y$. Since $\J_{\text{small}} = \J'$ is itself a database join, this is the desired sampler. This procedure is given in Algorithms \ref{alg:L2presample} (pre-processing step) and \ref{alg:L2sample} (sampling step), described in  Section \ref{sec:twotableanalysis}.
 We can apply this sampling primitive to efficiently sample from the generalized leverage scores in time substantially less than constructing $\J_{\text{small}}$, which ultimately allows for our final subspace embedding guarantee of Theorem \ref{thm:subspace}.

Finally, to obtain our input sparsity runtime machine precision regression algorithm, we apply our subspace embedding with constant $\eps$ to precondition the join $\J$, after which the regression problem can be solved quickly via gradient descent. While a general gradient step is not always possible to compute efficiently with respect to the join $\J$, we demonstrate that when the products used in the gradient step arise from vectors in the column span of $\J$, the updates can be computed efficiently, which will yield our main regression result (Theorem \ref{thm:reg}). 

\subsection{Analysis}\label{sec:twotableanalysis}

 We will begin by proving our main technical sampling result, which proceeds in a series of lemmas, and demonstrates that the construction of the diagonal sampling matrix $\S_{\text{small}}$ in Algorithm \ref{alg:1} can be carried out extremely quickly.  
 
 \begin{proposition}\label{prop:nbound}
 Let $\J_{\text{small}} \in \R^{n_{\text{small}} \times d}$ be the matrix constructed as in Algorithm \ref{alg:1}. Then we have $n_{\text{small}} \leq  (n_1 + n_2) d \cdot \gamma$, where $\gamma=1$ in the dense case and $\gamma =d$ in the sparse case.
 \end{proposition}
 
 \begin{proof}
Recall that $\J_{\text{small}}$ consists of all blocks $i$ of $\J$ with $\max \{s_{(i),1}, \allowbreak s_{(i),2} \} < d \cdot \gamma $, and thus $s_{(i)} \leq d^2 \gamma^2$.  
The total number $n_{\text{small}}$ of rows in  $\J_{\text{small}}$ is then $\sum_{i \in \mathcal{B}_{\text{small}} } s_{(i),1}\cdot  s_{(i),2} \leq \|s^1_{\text{small}}\|_2 \|s^2_{\text{small}}\|_2$ by the Cauchy-Schwarz inequality, where $s^j_{\text{small}}$ is the vector with coordinates given by the values $s_{(i),j}$ for $i \in \mathcal{B}_{\text{small}}$. 
Observe that these vectors admit the $\ell_1$ bound of $\|s^j_{\text{small}}\|_1 \leq n_j$ since each table $\T_j$ has only $n_j$ rows. Moreover, they admit the $\ell_\infty$ bound of $\|s^j_{\text{small}}\|_\infty \leq d \gamma$.
With these two constraints, it is standard that the $\ell_2$ norm is maximized by placing all of the $\ell_1$ mass on coordinates with value given by the $\ell_\infty$ bound. 
It follows that $\|s^j_{\text{small}}\|_2$ is maximized by having $n_i/(d\gamma)$ coordinates equal to $d \gamma$, giving $\|s^j_{\text{small}}\|_2^2 \leq  n_j d \gamma$ for $j \in [2]$, so $n_{\text{small}} \leq \|s^1_{\text{small}}\|_2 \|s^2_{\text{small}}\|_2 \leq (n_1 + n_2) d \gamma$ as required. 
 \end{proof}
 
 We now demonstrate how we can quickly $\ell_2$ sample rows from a join-vector or join-matrix product after input sparsity time pre-processing. This procedure is split into two algorithms, Algorithm \ref{alg:L2presample} and \ref{alg:L2sample}.  Algorithm \ref{alg:L2presample} is an input sparsity time pre-processing step, which given $\J \in \R^{n \times d}$ and $\Y \in \R^{d \times t}$, constructs several binary tree data structures.  Algorithm  \ref{alg:L2sample} then uses these data structures to sample a row of the product $\J \Y$ with probability proportional to its $\ell_2$ norm, in time $O( \log N)$.

The following lemma shows we can compute $\S\J_{\text{small}}$ in input sparsity time using Algorithm \ref{alg:L2presample} and \ref{alg:L2sample}. 

\begin{lemma} \label{lem:twotablemain} Set the value $\gamma = 1$ in the dense case, and $\gamma = d$ in the sparse case.
Let $\J_{\text{small}} \in \R^{n_{\text{small}} \times d}$ and let $\tilde{\J}_{\text{small}} \in \R^{m \times d}$ be the subset of rows of $\J_{\text{small}} $ constructed in Algorithm \ref{alg:1}, where $m = \Theta( (n_1 + n_2)/\gamma)$. Let $\S$ be the diagonal sampling matrix as constructed in Algorithm \ref{alg:1}.  Then with probability $1-1/d$, we have that $\S \J_{\text{small}}$ is an $\eps$-subspace embedding for $\J_{\text{small}}$. Moreover, $\S$ has at most $\tilde{O}(d^2 \gamma^2 /\eps^2 )$ non-zero entries, and 
$\S\J_{\text{small}}$ can be computed in time $\tilde{O}(\nnz(\T_1) + \nnz(\T_2) + d^3 \gamma^2 /\eps^2 )$.
\end{lemma}

We defer the proof of the Lemma to Section \ref{app:twotablemain}, and first show how our main results follow given Lemma \ref{lem:twotablemain}.

\smallskip

\textbf{Theorem  \ref{thm:subspace} }[In-Database Subspace Embedding] {\it 
	Suppose  $\J=\T_1 \Join \T_2 \in \R^{n \times d}$ is a join of two tables, where $\T_1 \in \R^{n_1 \times d_1},\T_2 \in \R^{n_2 \times d_2}$. Then Algorithm \ref{alg:1} outputs a sketching matrix $\S^* \in \R^{k \times n}$ with $k = \tilde{O}(d^2 \gamma^2 /\eps^2)$ (where $\gamma$ is chosen as in Lemma \ref{lem:twotablemain})  such that $\tilde{\J} = \S^* \J$ is an $\eps$-subspace embedding for $\J$, meaning
	\[\|\S^*\J x\|_2^2 = (1 \pm \eps)\|\J x\|_2^2 \]
	for all $x \in \R^{d}$ with probability at least $9/10$. The runtime to return $\S^* \J$ is the minimum of $\tilde{O}((n_1 + n_2)d/\eps^2  +  d^{3}/\eps^2  )$ and  $\tilde{O}((\nnz(\T_1) + \nnz(\T_2))/\eps^2 +d^5/\eps^2   )$.
}

\begin{proof}
Our algorithm partitions the rows of $\J$ into those from $\mathcal{B}_{\text{small}}$ and  $\mathcal{B}_{\text{big}}$, and outputs the result of stacking sketches for $\J_{\text{small}}$ and $\J^{(i)}$ for each $i \in \mathcal{B}_{\text{big}}$. Thus it suffices to show that each sketch $\S^i \J^{(i)}$ is a subspace embedding for $\J^{(i)}$, and $\S \J_{\text{small}}$ is a subspace embedding for $\J_{\text{small}}$. The latter holds by Lemma \ref{lem:twotablemain}, and the former follows directly from applying Lemma \ref{lem:fastTenSketch} and a union bound over the at most $O(n)$ such $i$. Since each such $\J^{(i)}$  can be written as $\J^{i} = \left( \hat{\T}_1^{(i)} \otimes \mathbf{1}^{s_{(i),2} } +  \mathbf{1}^{s_{(i),1} } \otimes   \hat{\T}_2^{(i)} \right) \in \R^{s_{(i)} \times d}$, the Fast Tensor-sketch lemma can be applied to $\J^{(i)}$ in time $\tilde{O}((\nnz(\T_1^{(i)}) + \nnz(\T_2^{(i)}) + d^2)/\eps^2)$. Note that $|\mathcal{B}_{\text{big}}| \leq 2(n_1 + n_2)/(d\gamma)$, since there can be at most this many values of $s_{(i)}^1  +s_{(i)}^2  > d \gamma$. Thus the total running time is $\tilde{O}(\sum_{i \in \mathcal{B}_{\text{big}}} ( \nnz(\T_1^{(i)}) + \nnz(\T_2^{(i)}) + d^2)/\eps^2)  =  \tilde{O}( (\nnz(\T_1) + \nnz(\T_2))/\eps^2 + (n_1 + n_2)d/(\gamma\eps^2))$. Finally, applying  CountSketch will cost $\tilde{O}((n_1 + n_2)/(d\gamma)\cdot d/\eps^2\cdot d)$, which is $\tilde{O}((n_1 + n_2)d/\eps^2)$ for the dense case and $\tilde{O}((n_1 + n_2)/\eps^2)$ for the sparse case. The remaining runtime analysis follows from Lemma \ref{lem:twotablemain}, setting $\gamma$ to be either $1$ or $d$. 

\end{proof}




We now demonstrate how our subspace embeddings can be easily applied to obtain fast algorithms for regression. To do this, we will first need the following proposition, which shows that $w^T \J^T \J$ can be computed in input sparsity time for any $w \in \R^d$. \\

\begin{proposition}\label{prop:fastprod}
	Suppose  $\J=\T_1 \Join \T_2 \in \R^{n \times d}$ is a join of two tables, where $\T_1 \in \R^{n_1 \times d_1},\T_2 \in \R^{n_2 \times d_2}$. Let $b \in \R^n$ be any column of the join $\J$. Let $U \subseteq [d]$ be any subset, and let $\J_U \in \R^{n \times |U|}$ be the subset of the columns of $\J$ contained in $\U$. Let $w \in \R^{d}$ be any vector, and let $x = \J w \in \R^n$. Then given $w$, the vector $x^T \J_{U}  = w^T \J^T \J_U \in \R^{|U|}$ can be computed in time $O(\nnz(\T_1)+ \nnz(\T_2))$. 
\end{proposition}
\begin{proof}
Fix any $i \in \mathcal{B}(\J)$, and similarly let $\J^{(i)}_{U},\hat{\T}_{1,U}^{(i)},\hat{\T}_{2,U}^{(i)} $ be $\J^{(i)},\hat{\T}_{1}^{(i)},\hat{\T}_{2}^{(i)}$ respectively, restricted to the columns of $U$. Note that we have $\J^{(i)}_{U} = (\hat{\T}_{1,U}^{(i)} \otimes \mathbf{1}^{s_{(i),2}} +  \mathbf{1}^{s_{(i),2}} \otimes   \hat{\T}_{2,U}^{(i)}) \allowbreak \in \R^{s_{i} \times |U|}$. 
Let $x^{(i)} \in \R^{s_{(i)}}$ be $x$ restricted to the rows inside of block $i \in \mathcal{B}(\J)$. 
Since $x = \J w$ for some $w \in \R^d$, we have $x^{(i)} = \J^{(i)} w = (\hat{\T}_{1}^{(i)} \otimes \mathbf{1}^{s_{(i),2}} +  \mathbf{1}^{s_{(i),1}} \otimes   \hat{\T}_{2}^{(i)})w$. 
Then we have
\begin{equation}
    \begin{split}
    x^T \J_{U} &= \sum_{i \in \mathcal{B}} (x^{(i)})^T \J^{(i)}_{U}\\
    &= \sum_{i \in \mathcal{B}}  w^T \left(\hat{\T}_{1}^{(i)} \otimes \mathbf{1}^{s_{(i),2}} +  \mathbf{1}^{s_{(i),1}} \otimes   \hat{\T}_{2}^{(i)}\right)^T \\ 
    & \cdot \left(\hat{\T}_{1,U}^{(i)} \otimes \mathbf{1}^{s_{(i),2}} +  \mathbf{1}^{s_{(i),1}} \otimes   \hat{\T}_{2,U}^{(i)} \right)\\ 
        &= \sum_{i \in \mathcal{B}}  w^T  \Big(s_{(i),2}(\hat{\T}_{1}^{(i)})^T \hat{\T}_{1,U}^{(i)}  + ((\hat{\T}_{1}^{(i)})^T \mathbf{1}^{s_{(i),1}}) \\
        & \otimes(  (\mathbf{1}^{s_{(i),2}})^T \hat{\T}_{2,U}^{(i)})+ (\mathbf{1}^{s_{(i),1}})^T \hat{\T}_{1,U}^{(i)} \otimes((\hat{\T}_{2}^{(i)})^T \mathbf{1}^{s_{(i),2}} )  \\
        & + s_{(i),1}(\hat{\T}_{2}^{(i)})^T \hat{\T}_{2,U}^{(i)}  \Big)\\ 
        \end{split}
\end{equation}
where the last equality follows from the mixed product property of Kronecker products (see e.g.,  \cite{van2000ubiquitous}). 
First note that the products $ w^T  s_{(i),2}(\hat{\T}_{1}^{(i)})^T \hat{\T}_{1,U}^{(i)}$ and $w^T s_{(i),1}(\hat{\T}_{2}^{(i)})^T  \hat{\T}_{2,U}^{(i)}$ can be computed in $O(\nnz(\T_1^{(i)}) + \nnz(\T_2^{(i)}))$ time by computing the vector matrix product first. Thus, it suffices to show how to compute $w^T ((\mathbf{1}^{s_{(i),1}})^T \allowbreak\hat{\T}_{1,U}^{(i)}) \otimes((\hat{\T}_{2}^{(i)})^T \mathbf{1}^{s_{(i),2}} )$ and $w^T ((\hat{\T}_{1}^{(i)})^T \mathbf{1}^{s_{(i),1}}) \otimes(  (\mathbf{1}^{s_{(i),2}})^T \hat{\T}_{2,U}^{(i)})$ quickly. By reshaping the Kronecker products \cite{van2000ubiquitous}, we have 
$w^T \allowbreak(\mathbf{1}^{s_{(i),1}})^T \hat{\T}_{1}^{(i)} \allowbreak \otimes((\hat{\T}_{2,U}^{(i)})^T \mathbf{1}^{s_{(i),2}} )  =  ((\hat{\T}_{2,U}^{(i)})^T \mathbf{1}^{s_{(i),2}}w^T (\hat{\T}_{1}^{(i)})^T \mathbf{1}^{s_{(i),1}} )^T$
Now $w^T (\hat{\T}_{1}^{(i)})^T$ can be computed in $O(\nnz(\T_1^{(i)}))$ time, at which point $w^T (\hat{\T}_{1}^{(i)})^T \mathbf{1}^{s_{(i),1}}$ can be computed in $O(s_{(i),1}) = \allowbreak O( \allowbreak\nnz(\T_1^{(i)}))$ time. Next, we can compute $\mathbf{1}^{s_{(i),2}}(w^T (\hat{\T}_{1}^{(i)})^T \allowbreak \mathbf{1}^{s_{(i),1}})$ in  $O(s_{(i),2})$ $= \allowbreak O(\allowbreak \nnz(\T_2^{ (i)}))$ time. Finally, $(\hat{\T}_{2,U}^{(i\allowbreak )})^T \allowbreak ( \mathbf{1}^{s_{(i),2}}w^T  \cdot \allowbreak (\hat{\T\allowbreak }_{1}^{(i)})^T \allowbreak \mathbf{1}^{s_{(i),1}})$  can be computed in $ O(\nnz(\T_2^{(i)}))$ time. A similar argument holds for computing $w^T (\mathbf{1}^{s_{(i),1}})^T \allowbreak\hat{\T}_{1,U}^{(i)} \allowbreak \otimes((\hat{\T}_{2}^{(i)})^T\allowbreak  \mathbf{1}^{s_{(i),2}} )$, which completes the proof, noting that $\sum_{i \in \mathcal{B}} \nnz(\T_1^{(i)}) + \nnz(\T_2^{(i)})  = \nnz(\T_1) + \nnz(\T_2)$.

\end{proof}

We now state our main theorem for machine precision regression. We remark again that the success probability can be boosted to $1-\delta$ by boosting the success probability of the corresponding subspace embedding to $1-\delta$, as described earlier.  \\

\textbf{Theorem \ref{thm:reg} } [In-Database Regression (Theorem \ref{thm:reg})] {\it  
	 Suppose  $\J\allowbreak =\T_1 \Join \T_2 \in \R^{n \times d}$ is a join of two tables, where $T_1 \in \R^{n_1 \times d_1}, T_2\allowbreak  \in \R^{n_2 \times d_2}$. Let $\U \subset [d]$ be any subset, and let $\J_U \in \R^{N \times |U|}$ be $\J$ restricted to the columns in $U$, and let $b \in \R^n$ be any column of the join $\J$. Then there is an algorithm which returns $\hat{x} \in \R^{|U|}$ such that with probability $9/10$ we have
	\[     \|\J_{U} \hat{x} - b\|_2 \leq (1+\eps) \min_{x \in \R^{|U|}}\|\J_{U} x - b\|_2 \]
	The runtime required to compute $\hat{x}$ is the minimum of \\ 
	$\tilde{O}\left( ((n_1 + n_2)d  +  d^{3}) + (\nnz(\T_1) + \nnz(\T_2) + d^2)\log(1/\eps) \right)$ and  \\$\tilde{O}\left( d^5 + (\nnz(\T_1) + \nnz(\T_2) + d^2)\log(1/\eps) \right)$.
}	

\begin{proof}
The following argument follows a standard reduction from having a subspace embedding to obtaining high precision regression (see, e.g., Section 2.6 of \cite{woodruff2014sketching}). 
We first compute a subspace embedding $\tilde{\J}_ = \S^* \J$ for $\J$ via Theorem \ref{thm:subspace} with precision parameter $\eps_0 = \Theta(1)$, so that $\S^* \J$ has $k= \tilde{O}(d^2 \gamma^2)$ rows. Note that in particular this implies that  $\S^* \J_{U}$ is an $\eps_0$-subspace embedding for $\J_U$. 
We then generate an OSNAP matrix $\W \in \R^{\tilde{\Theta}(d) \times k}$ via Lemma \ref{lem:OSNAP} with precision $\eps_0$, and condition on the fact that $\W \S^* \J_U$ is an $\eps_0$-subpsace embedding for $\S^*\J_U$, which holds with large constant probability, from which it follows that $\W \S^* \J_{U}$ is an $O(\eps_0)$-subspace embedding for $\J_{U}$. We then compute the QR factorization $\W \S^* \J_{U} = \Q\RR^{-1}$, which can be done in $O(d^{\omega})$ time via fast matrix multiplication \cite{demmel2007fast}. By standard arguments \cite{woodruff2014sketching}, the matrix $\J_{U} \RR$ is now $O(1)$-well conditioned -- namely, we have $\sigma_{\max}(\J_{U}\RR)/\sigma_{\min}(\J_{U}\RR) = O(1)$. 
Given this, we can apply the gradient descent update $x_{t+1} \leftarrow x_t + \RR^T \J^T_{U}(b - \J_{U} \RR x_t)$, which can be computed in $O(\nnz(\T_1) + \nnz(\T_2) + d^2)$ time via Proposition \ref{prop:fastprod} (note we compute $\RR x_t$ in $O(d^2)$ time first, and then compute $(\J_{U} \RR x_t)$. Here we use the fact that $b - \J_{U} \RR x_t = \J w$ for some $w$ which is a vector in the column span of $\J$, and moreover, we can determine the value of $w$ from computing $\RR x_t$ and noting that $b = \J e_{j^*}$, where $j^* \in[d]$ is the index of $b$ in $\J$. Since $ \RR^T \J^T_{U}$ is now well -conditioned, gradient descent now converges in $O(\log 1/\eps)$ iterations given that we have a constant factor approximation $x_0$ \cite{woodruff2014sketching}. Specifically, it suffices to have an $x_0 \in \R^d$ such that $\|\J_U x_0 - b\|_2 \leq (1+\eps_0) \min_{x \in \R^d} \| \J_U x - b\|_2$. But recall that such an  $x_0$ can be obtained by simply solving
$x_0 = \arg \min_x \|\S^*\J_{U} x - \S^*b\|_2$,
using the fact that $\S^*$ is an $\eps_0$- subspace embedding for the span of $\J$, which completes the proof of the theorem. 
\end{proof}

\subsection{Proof of Lemma \ref{lem:twotablemain}}\label{app:twotablemain}
\begin{algorithm}[h]
\caption{Pre-processing step for fast $\ell_2$ sampling from rows of $\J \cdot \Y$, where  $\J=\T_1 \Join \T_2$ and $\Y \in \R^{d \times r}$. } \label{alg:L2presample}
\begin{algorithmic}[1]
\STATE
		For each block $i \in \mathcal{B}$, compute $a_{(i)} = \hat{T}_1^{(i)} \Y$ and $b_{(i)} = \hat{T}_2^{(i)}$ 
		\STATE
			For each block $i \in \mathcal{B}$,  construct a binary tree $\tau_{(i)}(\T_1), \tau_{(i)}(\T_2)$ as follows:
			
		\STATE
			   Each node $v \in \tau_{(i)}(\T_j)$ is a vector $v \in \R^{3 r}$
	    \STATE
			   If $v$ is not a leaf, then  $v = v_{\text{lchild}} + v_{\text{rchild}}$ with $v_{\text{lchild}},v_{\text{rchild}}$ the left and right children of $v$.
	    \STATE
			    The leaves of $\tau_{(i)}(\T_j)$ are given by the set
			    \[     \{  v_{l}^{(i),j} =\left( v_{l,1}^{(i),j}, v_{l,2}^{(i),j}, \dots,v_{l,r}^{(i),j}  \right)\in \R^{3r} \; | \; l \in[ s_{(i),j}] \}\]
			   \noindent where $v_{l,q}^{(i),1} = (1,2(a_{(i)})_{ l,q}, (a_{(i)})_{ l,q}^2 ) \in \R^3$ and  $v_{l,q}^{(i),2} = ( (b_{(i)})_{ l,q}^2,(b_{(i)})_{ l,q},1 ) \in \R^3$. 
		\STATE
			Compute the values $\langle \text{root}(\tau_i(\T_1)) , \text{root}(\tau_i(\T_2))\rangle$ for all $i \in \mathcal{B}$, and also compute the sum $\sum_{i} \langle \text{root}(\tau_i(\T_1)) , \text{root}(\tau_i(\T_2))\rangle$.
				
				\end{algorithmic}
				\end{algorithm}
			\begin{algorithm}[h]
			\caption{Sampling step for fast $\ell_2$ sampling from rows of $\J \cdot \Y$. } \label{alg:L2sample}
            \begin{algorithmic}[1]
            \STATE
			Sample a block $i \in \mathcal{B}$ with probability \[\frac{\left\langle \text{root}(\tau_i(\T_1)) , \text{root}(\tau_i(\T_2))\right\rangle}{\sum_{j} \left\langle \text{root}(\tau_j(\T_1)) , \text{root}(\tau_j(\T_2))\right\rangle}\]
			\STATE
			Sample $ l_1 \in[ s_{(i),1}]$ with probability \[p_{l_1} = \frac{\left\langle v^{(i),1}_{l_1} , \text{root}(\tau_i(\T_2))\right\rangle}{\sum_{l} \left\langle  v^{(i),1}_{l} , \text{root}(\tau_j(\T_2))\right\rangle}\]	\label{line:l1}
			\STATE
				Sample $ l_2 \in[ s_{(i),2}]$ with probability \[p_{l_2 \; | \; l_1} = \frac{\left\langle v^{(i),1}_{l_1} , v^{(i),2}_{l_2} \right\rangle}{\sum_{l} \left\langle  v^{(i),1}_{l_1} , v^{(i),2}_{l} \right\rangle}\]	\label{line:l2}
			\STATE
	Return the row corresponding to $(l_1,l_2)$ in block $i$.
		
		\end{algorithmic}
		\end{algorithm}
We start our proof by showing Algorithm \ref{alg:L2sample} can sample one row with probability according to the $\ell_2$ norm quickly after running the pre-processing Algorithm \ref{alg:L2presample}.
\begin{lemma}\label{lem:sampfast}
 Let $\J = \T_1 \Join \T_2 \in \R^{N \times d}$ be any arbitrary join on two tables, with $\T_i \in \R^{n_i \times d_i}$, and fix any $\Y \in \R^{d \times t}$. Then Algorithms \ref{alg:L2presample} and \ref{alg:L2sample}, after an $O(( \nnz(\T_1) + \nnz(\T_2)) (t + \log N))$-time pre-processing step (Algorithm \ref{alg:L2presample}), can produce samples $i^* \sim \mathcal{D}_{\Y}$ (Algorithm \ref{alg:L2sample}) from the distribution $\mathcal{D}_{\Y}$ over $[N]$ given by 
     \[     \text{Pr}_{i^* \sim \mathcal{D}_{\Y}} \left[ i^* = j \right] = \frac{\|(\J \cdot \Y)_{j,*}\|_2^2}{ \|\J \cdot \Y\|_F^2   } \] 
     such that each sample is produced in $O(\log N)$ time. 
 \end{lemma}
 \begin{proof}
 We begin by arguing the correctness of Algorithms \ref{alg:L2presample} and \ref{alg:L2sample}. Let $j^*$ be any row of $\J \cdot \Y$. Note that the row $j^*$ corresponds to a unique block $i \in \mathcal{B} = \mathcal{B}(\J)$, and two rows $l_1 \in [s_{(i),1}],l_2 \in [s_{(i),2}]$, such that $\J_{j^*,*} = (\hat{T}_1)_{l_1',*} + (\hat{T}_2)_{l_2',*}$, where $l_1' \in [n_1], l_2' \in [n_2]$ are the indices which correspond to  $l_1,l_2$. For $l \in [s_{(i),j}]$, let $v_l^j, v^j_{l,q}$ be defined as in Algorithm \ref{alg:L2presample}. We first observe that if $j^*$ corresponds to the block $i \in \mathcal{B}$, then $\langle v^{(i),1}_{l_1} , v^{(i),2}_{l_2} \rangle = \|(\J \Y)_{j^*,*}\|_2^2$, since 
 \begin{equation}
     \begin{split}
     & \langle v^{(i),1}_{l_1} , v^{(i),2}_{l_2} \rangle\\
     &= \sum_{q=1}^r v^{(i),1}_{l_1,q} \cdot v^{(i),2}_{l_2,q}\\
         &= \sum_{q=1}^r (a_{(i)})_{l_1,q}^2 +  2  (a_{(i)})_{l_1,q} (b_{(i)})_{l_2,q}+ (b_{(i)})_{l_2,q}^2 \\
         & =\sum_{q=1}^r \left( (a_{(i)})_{l_1,q}+ (b_{(i)})_{l_2,q}\right)^2 \\
         & = \|(\J \Y)_{j^*,*}\|_2^2
     \end{split}
 \end{equation}
 where for each block $i \in \mathcal{B}$, we compute $a_{(i)} = \hat{T}_1^{(i)} \Y$ and $b_{(i)} = \hat{T}_2^{(i)}$ as defined in Algorithm \ref{alg:L2presample}. 
 Thus it suffices to sample a row $j^*$, indexed by the tuple $(i,l_1,l_2)$ where $i \in \mathcal{B}, l_1 \in [s_{(i),1}],l_2 \in [s_{(i),2}]$, such that the probability we sample $j^*$ is given by $p_{j^*} = \langle v^{(i),1}_{l_1} , v^{(i),2}_{l_2} \rangle/( \sum_{i',l_1',l_2'}\langle v^{(i'),1}_{l_1'} , v^{(i'),2}_{l_2'} \rangle  )$. 
 We argue that Algorithm \ref{alg:L2sample} does precisely this. First note that for any $i \in \mathcal{B}$, we have  
 \[\langle v^{(i),1}_{l_1} , \text{root}(\tau_i(\T_2))\rangle =\sum_{l_1\in [s_{(i),1}] ,l_2 \in [s_{(i),2}]}\langle v^{(i),1}_{l_1} , v^{(i),2}_{l_2} \rangle\] 
 Thus we first partition the set of all rows $j^*$ by sampling a block $i$ with probability $\frac{\left\langle \text{root}(\tau_i(\T_1)) , \text{root}(\tau_i(\T_2))\right\rangle}{\sum_{j} \left\langle \text{root}(\tau_j(\T_1)) , \text{root}(\tau_j(\T_2))\right\rangle}$, which is exactly the distribution over blocks induced by the $\ell_2$ mass of the blocks. 
 Conditioned on sampling $i \in \mathcal{B}$, it suffices now to sample $l_1,l_2$ from that block.
 To do this, we first sample $l_1$ with probability  $\langle v^{(i),1}_{l_1} , \text{root}(\tau_i(\T_2))\rangle/ \allowbreak (\sum_{l} \langle  v^{(i),1}_{l} , \text{root}(\tau_j(\T_2)) \allowbreak\rangle)$, which is precisely the distribution  over indices $l_1 \in [s_{(i),1}]$ induced by the contribution of $l_1$ to the total $\ell_2$ mass of block $i$. 
 Similarly, once conditioned on $l_1$, we sample $l_2$ with probability $\langle v^{(i),1}_{l_1} , v^{(i),2}_{l_2} \rangle/(\sum_{l} \langle  v^{(i),1}_{l_1} , v^{(i),2}_{l} \rangle)$, which is the distribution  over indices $l_2\in [s_{(i),2}]$ induced by the contribution of the row $(l_1,l_2)$, taken over all $l_2$ with $l_1$ fixed.
 Taken together, the resulting sample $j^* \cong (i,l_1,l_2)$ is drawn from precisely the desired distribution. 
 
 Finally, we bound the runtime of this procedure. First note that computing $a_{(i)} = \hat{T}_1^{(i)} \Y$ and $b_{(i)} = \hat{T}_2^{(i)}$ for all blocks $i$ can be done in $O(t (\nnz(\T_1) + \nnz(\T_2)))$ time, since each row of the tables $\T_1,\T_2$ is in exactly one of the blocks, and each row is multiplied by exactly $t$ columns of $\Y$. Once the $a_{(i)},b_{(i)}$ are computed, each tree $\tau_{(i)}(\T_j)$ can be computed bottom up in time $O(\log N)$, giving a total time of  $O( \log N(\nnz(\T_1) + \nnz(\T_2)))$ for all trees. Given this, the values $\left\langle \text{root}(\tau_i(\T_1)) , \text{root}(\tau_i(\T_2))\right\rangle$ can be computed in less than the above runtime. Thus the total pre-processing time is bounded by $O((t+ \log N) (\nnz(\T_1) + \nnz(\T_2)))$ as needed. 
For the sampling time, it then suffices to show that we can carry out Lines 2 and 3 in  $O(\log N)$ time. But these samples can be samples from the root down, by first computing $\langle \text{root}_{\text{lchild}}(\tau_i(\T_1)) , \allowbreak \text{root}(\tau_i(\T_2))\rangle$ and $\langle \text{root}_{\text{rchild}}(\tau_i(\T_1)) , \allowbreak \text{root}(\tau_i(\T_2))\rangle$, sampling one of the left or right children with probability proportional to its size, and recursing into that subtree. Similarly, $l_2$ can be sampled by first computing $\langle v^{(i),1}_{l_1} , \text{root}_{\text{lchild}}(\tau_i(\T_2)) \rangle$	and $\langle v^{(i),1}_{l_1} , \allowbreak \text{root}_{\text{rchild}}(\tau_i(\T_2)) \rangle$ sampling one of the left or right children with probability proportional to its size, and recursing into that subtree. This completes the proof of the $O( \log N)$ runtime for sampling after pre-processing has been completed. 
 \end{proof}
 
Then we show how we can construct $\S$ by invoking Algorithm \ref{alg:L2presample} and \ref{alg:L2sample}.
 
\begin{lemma}\label{lem:makeS}
Let $\J_{\text{small}} \in \R^{n_{\text{small}} \times d}$ be the matrix \allowbreak constructed as in Algorithm \ref{alg:1} in the dense case. Then the diagonal sampling matrix $\S$, as defined within lines 4 through 8 of Algorithm \ref{alg:1}, can be constructed in time $\tilde{O}(\nnz(\T_1) + \nnz(\T_2) + d^2\gamma^2 / \eps^2)$.
\end{lemma}
\begin{proof}

We first show how we can quickly construct the matrix $\tilde{\J}_{\text{small}}$, which consists of $m= \Theta((n_1 + n_2)/\gamma)$ uniform samples from the rows of $\J_{\text{small}}$. First, to sample the rows uniformly, since we already know the size of each block $s_{(i)}$, we can first sample a block $i \in \mathcal{B}_{\text{small}}$ with probability proportional to its size, which can be done in $O(\log(|\mathcal{B}_{\text{small}}|) )= O(\log N)$ time after the $s_{(i)}$'s are computed. Next, we can sample a row uniformly from $T_{i}^j$ for each $j \in [2]$, and output the join of the two chosen rows, the result of which is a truly uniform row from $\J_{\text{small}}$. Since we need $m$ samples, and each sample has $d$ columns, the overall runtime is $\tilde{O}((n_1 + n_2)d/\gamma)$ to construct $\tilde{\J}_{\text{small}}$, which is $O(\nnz(\T_1) + \nnz(\T_2) )$ in the sparse case.

Once we have $\tilde{\J}_{\text{small}}$, we compute in line \ref{line:4} of Algorithm \ref{alg:1} the sketch $\W \cdot  \tilde{\J}_{\text{small}}$, where $\W \in \R^{t \times d}$ is the OSNAP Transformation of Lemma \ref{lem:OSNAP} with $\eps = 1/100$, where $t = \tilde{O}(d)$, which we can compute in $\tilde{O}(md) = \tilde{O}((n_1 + n_2)d/\gamma)$ time by Lemma  \ref{lem:OSNAP}. Given this sketch $\W \cdot  \tilde{\J}_{\text{small}} \in \R^{t \times d}$, the SVD of the sketch can be computed in time $O(d^\omega)$ \cite{demmel2007fast}, where $\omega < 2.373$ is the exponent of fast matrix multiplication. Since $\W   \tilde{\J}_{\text{small}}$ is a $1/100$ subspace embedding for $\tilde{\J}$ with probability $99/100$ by Lemma \ref{lem:OSNAP}, by Proposition \ref{prop:tausubspace} we have $\tau^{\W   \tilde{\J}_{\text{small}}}(\AA) = (1 \pm 1/100)\tau^{\tilde{\J}_{\text{small}}}(\AA)$ for any matrix $\AA$. Next, we can compute $\V \Sigma\G$ in the same $O(d^\omega)$ runtime, where $\G \in \R^{d \times t }$ is a Gaussian matrix with $t= \Theta( \log N)$ and with entries drawn independently from $\mathcal{N}(0,1/t^2)$. By standard arguments for Johnson Lindenstrauss random projections (see, e.g., Lemma 4.5 of \cite{li2013iterative}), we have that $\|( x^T \G)\|_2^2 =(1 \pm 1/100) \|x\|_2^2$ for any fixed vector $x \in \R^d$ with probability at least $1-n^{-c}$ for any constant $c \geq 1$ (depending on $t$). 

We now claim that $\tilde{\tau_i}$ as defined in Algorithm \ref{alg:1} satisfies $C^{-1}  \allowbreak\tau_i^{\tilde{\J}_{\text{small}}} (\J_{\text{small}}) \leq \tilde{\tau_i} \leq C\tau_i^{\tilde{\J}_{\text{small}}} (\J_{\text{small}}) $ for some fixed constant $C \geq 1$. 
As noted above, $\tau_i^{\tilde{\J}_{\text{small}}} (\J_{\text{small}}) = (1 \pm 1/100)\tau^{\W   \tilde{\J}_{\text{small}}}(\AA)$, so it suffices to show $C^{-1} \tau_i^{\W \tilde{\J}_{\text{small}}} (\J_{\text{small}}) \allowbreak \leq \tilde{\tau_i} \allowbreak \leq C\tau_i^{\W\tilde{\J}_{\text{small}}} (\J_{\text{small}})$.
To see this, first note that if $(\J_{\text{small}})_{i,*}$ is contained within the row span of $\tilde{\J}_{\text{small}}$, then
\begin{equation}
    \begin{split}
        \tau_i^{\tilde{\J}_{\text{small}}} (\J_{\text{small}}) &=  \|(\J_{\text{small}})_{i,*} \V \Sigma^{-1}\|_2^2  \\
        &= (1 \pm 1/100) \|(\J_{\text{small}})_{i,*} \V \Sigma^{-1} \G\|_2^2 \\
        &=  (1 \pm 1/100) \tilde{\tau_i}
    \end{split}
\end{equation}
Thus it suffices to show that if $(\J_{\text{small}})_{i,*}$ has a component outside of the span of $\tilde{J}_{\text{small}}$, then we have $\|(\J_{\text{small}})_{i,*} \allowbreak (\mathbb{I}_{d} - \V \V^T ) g \|_2^2 > 0$. 
To see this, note that $\left(\mathbb{I}_{d} - \V \V^T \right)$ is the projection onto the orthogonal space to the span of  $\tilde{\J}_{\text{small}}$. Thus $ \left(\mathbb{I}_{d} - \V \V^T \right) g$ is a random non-zero vector in the orthogonal space of $\tilde{\J}_{\text{small}}$, thus $(\J_{\text{small}})_{i,*} \left(\mathbb{I}_{d} - \V \V^T \right) g \allowbreak \neq 0$ almost surely if  $(\J_{\text{small}})_{i,*}$ has a component outside of the span of $\tilde{\J}_{\text{small}}$, which completes the proof of the claim.


Finally, and most significantly, we show how to implement line 8 
of Algorithm \ref{alg:1}, which carries out the the construction of  $\S$. Given that $1/C  \tau_i^{\tilde{\J}_{\text{small}}} (\J_{\text{small}}) \leq \tilde{\tau_i} \leq C\tau_i^{\tilde{\J}_{\text{small}}} (\J_{\text{small}}) $, we can apply Theorem 1 of \cite{cohen2015uniform}, using that $n_{\text{small}} \leq (n_1 + n_2) d \cdot \gamma$ via Proposition \ref{prop:nbound}, which yields that $\sum_i  \tilde{\tau_i}   = O(n_{\text{small}}\frac{d}{m}) \allowbreak =   O(d^2 \gamma^2)$. Thus to construct the sampling matrix $\S$, it suffices to sample $\alpha = O(d^2 \gamma^2 \log d/\eps^2)$ samples from the distribution over the rows $i$ of $\J_{\text{small}}$ given by  $q_i = \frac{\tilde{\tau}_i}{\sum_i  \tilde{\tau_i} }$. We now describe how to accomplish this. 

We first show how to sample the rows $i$ with $(\J_{\text{small}})_{i,*} \allowbreak (\mathbb{I}_{d} - \V \V^T ) g = 0$.  To do this, it suffices to sample rows from the distribution induced by the $\ell_2$ norm of the rows of $(\J_{\text{small}}) \V \Sigma^{-1} \G$. To do this, we can simply apply Algorithms \ref{alg:L2presample} and \ref{alg:L2sample} to the product $\J_{\text{small}} \cdot ( \V \Sigma^{-1} \G)$. First note that we can do this because $\J_{\text{small}}$ itself is a join of $(\T_1)_{\text{small}}$ and $(\T_2)_{\text{small}}$, which are just $\T_1,\T_2$ with all the rows contained in blocks $i \in \mathcal{B}_{\text{big}}$ removed. Since $\V \Sigma^{-1} \G \in \R^{d \times t}$ for $t = \tilde{O}(1)$, by Lemma \ref{lem:sampfast} after $\tilde{O}(\nnz(\T_1) + \nnz(\T_2))$ time, for any $s \geq 1$ we can sample $s$ times independently from this induced $\ell_2$ distribution over rows in time $\tilde{O}(s)$. Altogether, we obtain the required $\alpha = O(d^2 \gamma^2 \log d/\eps^2)$ samples from the distribution over the rows $i$ of $\J_{\text{small}}$ given by  $q_i = \frac{\tilde{\tau}_i}{\sum_i  \tilde{\tau_i} }$ in the case that $(\J_{\text{small}})_{i,*} \left(\mathbb{I}_{d} - \V \V^T \right) g = 0$, with total runtime $\tilde{O}(\nnz(\T_1) + \nnz(\T_2) + d^2 \gamma^2  /{\eps^2})$.

Finally, for the case that $(\J_{\text{small}})_{i,*} \left(\mathbb{I}_{d} - \V \V^T \right) g > 0$, we can apply the same Algorithms \ref{alg:L2presample} and \ref{alg:L2sample} to sample from the rows of $\J_{\text{small}} \cdot \left(\mathbb{I}_{d} - \V \V^T \right) g$. First observe, using the fact that $\sum_i \tilde{\tau}_i \leq O(d^2 \gamma^2)$ by Theorem 1 of \cite{cohen2015uniform}, it follows that there are at most $O(d^2 \gamma^2)$ indices $i$ such that $\J_{\text{small}} \cdot \left(\mathbb{I}_{d} - \V \V^T \right) g > 0$. Thus, when applying Algorithm \ref{alg:L2sample} after the pre-processing step is completed, instead of sampling independently from the distribution induced by the norms of the rows, we can deterministically find all rows with $\J_{\text{small}} \cdot \left(\mathbb{I}_{d} - \V \V^T \right) g > 0$ in $\tilde{O}( d^2 \gamma^2)$ time, by simply enumerating over all computation paths of  Algorithm \ref{alg:L2sample} that occur with non-zero probability. Since there are $O( d^2 \gamma^2)$ such paths, and each one is carried out in $O(\log N)$ time by Lemma \ref{lem:sampfast}, the resulting runtime is the same as the case where $(\J_{\text{small}})_{i,*} \left(\mathbb{I}_{d} - \V \V^T \right) g = 0$.

Finally, we argue that we can compute exactly the probabilities $p_j$ with which we sampled a row $j$ of $\J_{\text{small}}$, for each $i$ that was sampled, which will be needed to determine the scalings of the rows of $\J_{\text{small}}$ that are sampled. For all the rows sampled with $(\J_{\text{small}})_{j,*} \left(\mathbb{I}_{d} - \V \V^T \right) g >0$, the corresponding value of $p_j$ is by definition $1$. Note that if a row was sampled in both of the above cases, then it should in fact have been sampled in the case that $(\J_{\text{small}})_{j,*} \left(\mathbb{I}_{d} - \V \V^T \right) g \allowbreak  >0$, so we set $p_j = 1$. For every row $j$ sampled via Algorithm \ref{alg:L2sample} when $(\J_{\text{small}})_{j,*} \left(\mathbb{I}_{d} - \V \V^T \right) g =0$, such that $j$ corresponds to the tuple $(i,l_1,l_2)$ where $i \in \mathcal{B}_{\text{small}}$ and $l_1 \in [s_{(i),1}], l_2 \in [s_{(i),2}]$, we can compute the probability it was sampled exactly via
 \[p_j = p_{(i,l_1,l_2)} = \frac{\langle v^{(i),1}_{l_1} , v^{(i),2}_{l_2} \rangle}{\sum_{c \in \mathcal{B}} \langle \text{root}(\tau_c(\T_1)) , \text{root}(\tau_c(\T_2))\rangle }\] using the notation in Algorithm \ref{alg:L2sample}. Then setting $\S_{j,j} = \frac{1}{\sqrt{p_j}}$ for each sampled row $j$ between both processes yields the desired construction of $\S$.
\end{proof}

With the above lemma in hand, we can give the proof of Lemma \ref{lem:twotablemain}.

\begin{proof}[Proof of lemma \ref{lem:twotablemain}]
As argued in the proof of Lemma \ref{lem:makeS}, we have that $(1/C)\tau_i^{\tilde{\J}_{\text{small}}} (\J_{\text{small}}) \leq \tilde{\tau_i} \leq C\tau_i^{\tilde{\J}_{\text{small}}} (\J_{\text{small}}) $ for some fixed constant $C \geq 1$ and all $i \in \mathcal{B}_{\text{small}}$. The result then follows immediately from Theorem 4 of \cite{cohen2015uniform}, using the fact that $n_{\text{small}}/m = \tilde{O}(d \gamma^2)$, where $m$ is the number of rows subsampled in $\tilde{\J}_{\text{small}}$ from $\J_{\text{small}}$. 
\end{proof}

%% file: Gneral.tex
\section{General Join Queries}\label{sec:general}

\label{section:general}
Below we introduce DB-Sketch as a class of algorithms, and show how any oblivious sketching algorithm that has the properties of DB-Sketch can be implemented efficiently for data coming from a join query. Since the required properties are very similar to the properties of linear sketches for Kronecker products, we will be able to implement them inside of a database. Given that the statistical dimension can be much smaller than the actual dimensions of the input data, our time complexity for ridge regression can be significantly smaller than that for ordinary least squares regression, which is important in the context of joins of many tables.

In the following we assume the join query is acyclic; nevertheless, for cyclic queries it is possible to obtain the hypergraph tree decomposition of the join and create a table for each vertex in the tree decomposition by joining the input tables that are a subset of the vertex's bag in the hypergraph tree decomposition. One can then replace the cyclic join query with an acyclic query using the new tables. 

In our algorithm, we use FAQ and inside-out algorithms as subroutines. The definition of FAQ is given in Appendix \ref{appendix:database:background}. Let $\J'=\T_1 \Join \T_2 \Join \dots \Join \T_m$ be an acyclic join. Let $\rho$ be a binary expression tree that shows in what order the algorithm inside-out \cite{abo2016faq} multiplies the factors for an arbitrary single-semiring FAQ; meaning, $\rho$ has a leaf for each factor (each table) and $(m-1)$ internal nodes such that if two nodes have the same parent then their values are multiplied together during the execution of inside-out. We number the tables based on the order that a depth-first-search visits them in this expression tree. We let $\rho$ denote the multiplication order of $\J'$.

Now that the ordering of the tables is fixed, we reformulate the Join table $\J'$ so that it can be expressed as a summation of tensor products. Assign each column $c$ to one of the input tables that has $c$, and then let $E_i$ denote the columns assigned to table $\T_i$, $X_i$ be the projection of $X$ onto $E_i$, and $D_i$ be the domain of the tuples $X_i$ (projection of $\T_i$ onto $E_i$).

Letting $N=|D_1||D_2| \dots |D_m|$, we can reformulate $\J'$ as $\J\in \mathbb{R}^{N\times d}$ to have a row for any possible tuple $X \in D_1 \times \dots \times D_m$. If a tuple $X$ is present in the join, we put its value in the row corresponding to it, and if it is not present we put $0$ in that row. Note that $\J$ has all the rows in $\J'$ and also may have many zero rows; however, we do not need to represent $\J$ explicitly, and the sparsity of $\J$ does not cause a problem. One key property of this  formulation is that by knowing the values of a tuple $x$, the location of $x$ in $\J$ is well-defined. Also note that since we have only added rows that are $0$, any subspace embedding of $\J$ would be a subspace embedding of $\J'$, and for all vectors $x$, $\norm{\J x}_2^2 = \norm{\J'x}_2^2$. 


Given a join query $\J$ with $m$ tables and its multiplication order $\rho$, an oblivious sketching algorithm is an $(m,\rho)$-DB-Sketch if there exists a function $F:\mathbb{R}^{n \times d} \to R_F$ where $R_F$ is the range of $F$ and $F(\AA)$ represents the sketch of $\AA$ in some form and has the following properties:
\begin{enumerate}
    \item $F(\AA_1 + \AA_2) = F(\AA_1) \oplus F(\AA_2)$ where $\oplus$ is a commutative and associative operator.
    \item For any $V$ resulting from a Kronecker product of matrices $\AA = \AA_1 \otimes \AA_2 \otimes \dots \otimes \AA_m \in \mathbb{R}^{n_1 \dots n_m}$, $F(\AA) = F_1(\AA_1) \odot F_2(\AA_2) \odot \dots \odot F_m(\AA_m)$ where $\odot$ is applied based on the ordering in $\rho$, and for all $i$, $F_i$ has the same range as $F$, and $F_i(X_1 + X_2) = F_i(X_1) \oplus F_i(X_2)$. Furthermore, it should be possible to evaluate $F_i(v_i)$ in time $O(T_f \nnz(v_i))$.
    \item For any values $A,B,C$ in the range of $F$, $A \odot (B \oplus C) = (A \odot B) \oplus (A \odot C)$
\end{enumerate}

\begin{theorem}
\label{thm:dbsketch}
Given a join query $\J'=\T_1 \Join \T_2 \Join \dots \Join \T_m$ and a DB-Sketch algorithm, there exists an algorithm to evaluate $F(\J')$ in time $O(m (T_{\odot} + T_{\oplus}) T_{\text{FAQ}} + T_f mnd)$ where $n$ is the size of the largest table and $T_{\text{FAQ}}$ is the time complexity of running a single semiring FAQ, while $T_{\odot}$ and $T_{\oplus}$ are the time complexities of $\odot$ and $\otimes$,  respectively.
\end{theorem}

\begin{corollary}
\label{corollary:general:join}
For any join query $\J$ with multiplication order $\rho$ of depth $m-1$ and any scalar $\lambda$, let $d_\lambda$ be the $\lambda$-statistical dimension of $\J$. Then there exists an algorithm that produces $\S \in \R^{k\times n}$ where $k=O(d_\lambda m^4/\epsilon^2)$ in time $O((mkd)T_{\text{FAQ}} + m^6 n d/\epsilon^2)$ such that with probability $1-\frac{1}{\poly(n)}$ simultaneously for all $x\in \mathbb{R}^d$
\[     \|\S\J x\|_2^2 + \lambda \|x\|_2^2 = (1 \pm \eps)(\|\J x\|_2^2 + \lambda \|x\|_2^2).\] 
\end{corollary}
\begin{proof}
    The proof follows by showing that the algorithm in Lemma \ref{lem:fastTenSketch} is a DB-Sketch. We demonstrate this by introducing the functions $F_i$ and the operators $\oplus$ and $\odot$.
    The function $F_i(v_i)$ is an OSNAP transform (Lemma \ref{lem:OSNAP}) of $v_i$, $\AA \odot \B$ is $\S(\AA \otimes \B)$ where $\S$ is the Tensor Subsampled Randomized Hadamard Transform as defined in \cite{ahle2020oblivious}, and $\AA \oplus \B$ is a summation of tensors $\AA$ and $\B$. Then it is easy to see that all the properties hold since Kronecker product distributes over summation.
    
    Since $F_i$ needs to be calculated for all rows in each table $\T_i$, which takes $O(T_f mnd)$ time, the operator $\oplus$ takes $O(kd)$ time to apply since the size of the sketch is $k \times d$. The operator $\otimes$ takes at most $O(kd)$ time to apply using the Fast Fourier Transform (FFT)  \cite{ahle2020oblivious}. Therefore, the total time complexity can be bounded by $O((mkd)T_{\text{FAQ}} + m^6 n d/\epsilon^2)$.
    
    Lastly, we remark that the algorithm in \cite{ahle2020oblivious} requires that the $\odot$ operator be applied to the input tensors in a binary fashion; however, it is shown in a separate version \cite{ahle2019almost} of the paper \cite{ahle2020oblivious} that the sketching construction and results of \cite{ahle2020oblivious} continue to hold when the tensor sketch is applied linearly. See Lemma 6 and 7 in \cite{ahle2019almost} and Lemma 10 in \cite{ahle2020oblivious}.
\end{proof}

Corollary \ref{corollary:general:join} gives an algorithm for all join queries when the corresponding hypertree decomposition is a path, or has a vertex for which all other vertices are connected to it. Although the time complexity for obtaining an $\epsilon$-subspace embedding ($\lambda=0$) is not better compared to the exact algorithm for ordinary least squares regression, for ridge regression it is possible to create sketches with many fewer rows and still obtain a reasonable approximation. 

In the following we explain the algorithm for DB-Sketch using the FAQ formulation and inside-out algorithm \cite{abo2016faq}. Let $\J_i$ denote the matrix resulting from keeping the column $i$ of $\J$ and replacing all other columns with $0$. Then we have $\J = \sum_i \J_i$. Using $\J_i$ we can define the proposed algorithm as finding $F(\J_i)$ for all tables $\T_i$ and then calculating $\oplus_i F(\J_i)$.
Therefore, all we need to do is to calculate $F(\J_i)$ for all values of $i$. In the following we introduce an algorithm for the calculation of $F(\J_i)$ and then $F(\J)$ is just the summation of the results for the different tables. 

Let $e(X_k)$ be the $D_k$-dimensional unit vector that is $1$ in the row corresponding to $X_k$, and let $v(X_k)$ be the $D_k \times d$ dimensional matrix that agrees with $X_k$ in the row corresponding to $X_k$, and is $0$ everywhere else.

\begin{lemma}
\label{lemma:matrix:and:kronocker:product}
     For all tables, 
     $\J_i = \sum_{X \in \J'} e(X_1) \otimes \dots \otimes e(X_{i-1}) \otimes v(X_i) \otimes e(X_{i+1}) \otimes \dots \otimes e(X_{m})$, where $\otimes$ is the Kronecker product.
\end{lemma}

\begin{proof}
    For each tuple $X$, the term inside the summation has $N=|D_1||D_2|\dots|D_m|$ rows and only $1$ non-zero row because all of the tensors have only one non-zero row. The non-zero row is the row corresponding to $X$, and its value is the value of the same row in $\J_i$; therefore, $\J_i$ can be obtained by summing over all the tuples of $\J$.
\end{proof}

\begin{lemma}
    Let $F$ be a DB-Sketch. Then $F(\J_i)$ can be computed in time $O((T_{\odot} + T_{\oplus}) T_{\text{FAQ}} + T_f ndm)$
\end{lemma}
\begin{proof}
    We define a FAQ for $F(\J_i)$ and then show how to calculate the result. Let $g_j(X_j) = F_j(e(X_j))$ for all $j \neq i$ and $g_i(X_i) = F_i(v(X_i))$. Note that the number of non-zero entries in $e(X_i)$ is at most $d$. Therefore, it is possible to find all values of $F_i(v(X_i))$ in time $O(T_f nd)$.
    The claim is it is possible to use the inside-out \cite{abo2016faq} algorithm for the following query and find $\S\J_i$:
    \begin{align*}
    \label{equality:query}
        \bigoplus_{X \in \J'} \bigodot_{j} g_j(X_j)
    \end{align*}
    The mentioned query would be a FAQ if $\bigodot$ were commutative and associative. However, since we defined the ordering of the tables based on the multiplication order $\rho$, the inside-out algorithm multiplies the factors exactly in the same order needed for the DB-sketch algorithm. Therefore, we do not need the commutative and associative property of the $\odot$ operator to run inside-out.
    
    Now we need to show that the query truly calculates $F(\J_i)$. Based on Lemma \ref{lemma:matrix:and:kronocker:product} and the properties of $F$ we have:
   
    \begin{flalign*}
        F(\J_i)         =& 
        F\Big(\sum_{X \in \J'} e(X_1) \otimes \dots 
        \\
        &\ \ \ \otimes e(X_{i-1}) \otimes v(X_i) \otimes e(X_{i+1}) \otimes \dots 
        \\
        &\ \ \ \otimes e(X_{m})\Big)
        \\
        =& \bigoplus_{X \in \J'}F(e(X_1) \otimes \dots
        \\
        &\ \ \ \otimes e(X_{i-1}) \otimes v(X_i) \otimes e(X_{i+1}) \otimes \dots
        \\
        &\ \ \ \otimes e(X_{m}))
        \\
        =&
        \bigoplus_{X \in \J'}\Big(F_1(e(X_1)) \odot \dots 
        \\
        &\ \ \ \odot F_{i-1}(e(X_{i-1})) \odot F_{i}(v(X_i)) \odot F_{i+1}(e(X_{i+1}))
        \\
        &\ \ \  \odot \dots  \odot F_m(e(X_{m}))\Big)
        \\
        =&
        \bigoplus_{X \in \J'} \bigodot_{j} g_j(X_j)
    \end{flalign*}
    as needed.
\end{proof}


\begin{proof}[Proof of Theorem \ref{thm:dbsketch}]
    Finding the values of $F_i(X_i)$ for all $i$ and all tuples of $X_i$ takes $O(T_f ndm)$ time because $F_i$ can be calculated in time $O(T_f \nnz(\T_i))$, and the total number of non-zero entries can be bounded by $O(ndm)$. The calculation of $F(\J_i)$ can be done in time $O(m (T_{\odot} + T_{\oplus}) T_{\text{FAQ}} + T_f ndm)$ using $m$ rounds of the inside-out algorithm where $T_{FAQ} = O(md^2n^{\text{fhtw}} \log(n))$ \cite{abo2016faq}. After this step, we need to aggregate the results using the $\oplus$ operator to obtain the final result which takes $O(m T_\oplus)$ time.
    \end{proof}

%% file: Experiments.tex
We study the performance of our sketching method on several real datasets, both for two-table joins and general joins.\footnote{Code available at \url{https://github.com/AnonymousFireman/ICML_code}} We first introduce the datasets we use in the experiments. We consider two datasets: LastFM \cite{Cantador:RecSys2011} and MovieLens  \cite{10.1145/2827872}. Both of them contain several relational tables. We will compare our algorithm with the FAQ-based algorithm on the joins of some relations.

The LastFM dataset has three relations: \textbf{Userfriends} (the friend relations between users), \textbf{Userartists} (the artists listened by each user) and \textbf{Usertaggedartiststimestamps} (the tag assignments of artists provided by each particular user along with the timestamps).

The MovieLens dataset also has three relations: \textbf{Ratings} (the ratings of movies given by the users and the timestamps), \textbf{Users} (gender, age, occupation, and zip code information of each user), \textbf{Movies} (release year and the category of each movie).

\input{plot_together}
\subsection{Two-Table Joins}
In the experiments for two-table joins, we solve the regression problem $\min_x \Vert\J_U x - b\Vert_2^2$,
where $\J = T_1 \Join T_2 \in \mathbb{R}^{N\times d}$ is a join of two tables, $U\subset [d]$ and $b$ is one of the columns of $\J$. In our experiments, suppose column $p$ is the column we want to predict. We will set $U = [d] \setminus \{p\}$ and $b$ to be the $p$-th column of $\J$.

To solve the regression problem, the FAQ-based algorithm computes the covariance matrix $\J^T\J$ by running the FAQ algorithm for every two columns, and then solves the normal equations $\J^T\J x = \J^T b$. Our algorithm will compute a subspace embedding $\Tilde{\J}$, and then solve the regression problem $\min_x \Vert\Tilde{\J}_Ux - \Tilde{b}\Vert_2^2$, i.e., solve $\tilde{\J}^T\tilde{\J}x = \tilde{\J}^T \tilde{b}$.

We compare our algorithm to the FAQ-based algorithm on the LastFM and MovieLens datasets. The FAQ-based algorithm employs the FAQ-based algorithm to calculate each entry in $\J^T \J$.

For the LastFM dataset, we consider the join of \textbf{Userartists} and \textbf{Usertaggedartiststimestamps}:
$\J_1 = \mathbf{UA}\Join_{\mathbf{UA}.\text{user} = \mathbf{UTA}.\text{user}}\mathbf{UTA}.$
Our regression task is to predict how often a user listens to an artist based on the tags.
For the MovieLens dataset, we consider the join of \textbf{Ratings} and \textbf{Movies}:
$\J_2 = \mathbf{R}\Join_{\mathbf{R}.\text{movie} = \mathbf{M}.\text{movie}}\mathbf{M}.$
Our regression task is to predict the rating that a user gives to a movie.

In our experiments, we do the dataset preparation mentioned in \cite{SystemF}, to normalize the values in each column to range $[0,1]$. For each column, let $v_{\max}$ and $v_{\min}$ denote the maximum value and minimum value in this column. We normalize each value $v$ to $(v - v_{\min})/(v_{\max} - v_{\min})$.

\subsection{General Joins}
For general joins, we consider the ridge regression problem. Specifically, our goal is to find a vector $x$ that minimizes $\Vert \J x - b\Vert_2^2 + \lambda \Vert x\Vert_2^2$, where $\J=\T_1\Join \cdots\Join \T_m \in \mathbb{R}^{N\times d}$ is an arbitrary join, $b$ is one of the columns of $\J$ and $\lambda>0$ is the regularization parameter. The optimal solution to the ridge regression problem can be found by solving the normal equations $(\J^T \J + \lambda \mathbb{I}_{d})x = \J^Tb$.

The FAQ-based algorithm is the same as in the experiment for two-table joins. It directly runs the FAQ algorithm a total of $d(d+1)/2$ times to compute every entry of $\J^T \J$.

We run the algorithm described in Section \ref{section:general}, as well as the FAQ-based algorithm, on the MovieLens-25m dataset, which is the largest of the MovieLens\cite{10.1145/2827872} datasets. We consider the join of \textbf{Ratings}, \textbf{Users} and \textbf{Movies}:
$\J_3 = \mathbf{R}\Join_{\mathbf{R}.\text{user}=\mathbf{U}.\text{user}}\mathbf{U}\Join_{\mathbf{U}.\text{movie}=\mathbf{M}.\text{movie}}\mathbf{M}.$
Our regression task is to predict the rating that a user gives to a movie.

\subsection{Results}
\begin{table}\centering
\caption{Experimental Results for Two-Table Joins}
\label{tab:Experiment_two-table}
\begin{tabular}{|c|c|c|c|c|c|c|}
\hline
     & $n_1$ & $n_2$ & $d$ &$T_{\text{bf}}$ & $T_{\text{ours}}$ & \text{err}\\
\hline
    $\J_{1}$ & 92834 & 186479 & 6 & .034 & .011 & 0.70\% \\
\hline
    $\J_{2}$ & 1000209 & 3883 & 23 & .820 & .088 & 0.66\% \\
\hline
\end{tabular}
\end{table}

We run the FAQ-based algorithm and our algorithm on those joins and compare their running times. To measure accuracy, we compute the relative mean-squared error, given by:	
\[ \text{err} =     \frac{\Vert \J_U x_{\text{ours}} - b\Vert_2^2 - \Vert\J_U x_{\text{bf}} - b\Vert_2^2}{\Vert\J_U x_{\text{bf}} - b\Vert_2^2}\]
in the experiments for two-table joins, where $x_{\text{bf}}$ is the solution given by the FAQ-based algorithm, and $x_{\text{ours}}$ is the solution given by our algorithm. All results (runtime, accuracy) are averaged over $5$ runs of each algorithm.

In our implementation, we adjust the target dimension in our sketching algorithm for each experiment, as in practice it appears unnecessary to parameterize according to the worst-case theoretical bounds when the number of features is small, as in our experiments. Additionally, for two-table joins, we replace the Fast Tensor-Sketch with Tensor-Sketch (\cite{ahle2020oblivious, pagh2013compressed}) for the same reason. The implementation is written in MATLAB and run on an Intel Core i7-7500U CPU with 8GB of memory.

We let $T_{\text{bf}}$ be the running time of the FAQ-based algorithm and $T_{\text{ours}}$ be the running time of our approach, measured in seconds. Table \ref{tab:Experiment_two-table} shows the results of our experiments for two-table joins. From that we can see our approach can give a solution with relative error less than $1\%$, and its running time is significantly less than that of the FAQ-based algorithm. 

For general joins, due to the size of the dataset, we implement our algorithm in Taichi \cite{hu2019taichi, hu2019difftaichi} and run it on an Nvidia GTX1080Ti GPU. We split the dataset into a training set and a validation set, run the regression on the training set and measure the MSE (mean squared error) on the validation set. We fix the target dimension and try different values of $\lambda$ to see which value achieves the best MSE.

Our algorithm runs in 0.303s while the FAQ-based algorithm runs in 0.988s. The relative error of MSE (namely, $\frac{\text{MSE}_{\text{ours}} - \text{MSE}_{\text{bf}}}{\text{MSE}_{\text{bf}}}$, both measured under the optimal $\lambda$) is only $0.28\%$. 

We plot the MSE vs. $\lambda$ curve for the FAQ-based algorithm and our algorithm in Figure \ref{fig:plot_sketch} and \ref{fig:plot_brute_force}. We observe that the optimal choice of $\lambda$ is much larger in the sketched problem than in the original problem. This is because the statistical dimension $d_{\lambda}$ decreases as $\lambda$ increases. Since we fix the target dimension, $\eps$ thus decreases. So a larger $\lambda$ can give a better approximate solution, yielding a better MSE even if it is not the best choice in the unsketched problem.

We also plot the relative error of the objective function in Figure \ref{fig:plot_obj}. For ridge regression it becomes
\[\frac{\left(\Vert \J x_{\text{ours}} - b\Vert_2^2 + \lambda \Vert x_{\text{ours}}\Vert_2^2\right)- \left(\Vert\J x_{\text{bf}} - b\Vert_2^2 + \lambda \Vert x_{\text{bf}}\Vert_2^2\right)}{\Vert\J x_{\text{bf}} - b\Vert_2^2 + \lambda \Vert x_{\text{bf}}\Vert_2^2}.\] We can see that the relative error decreases as $\lambda$ increases in accordance with our theoretical analysis.



%% file: plot_together.tex
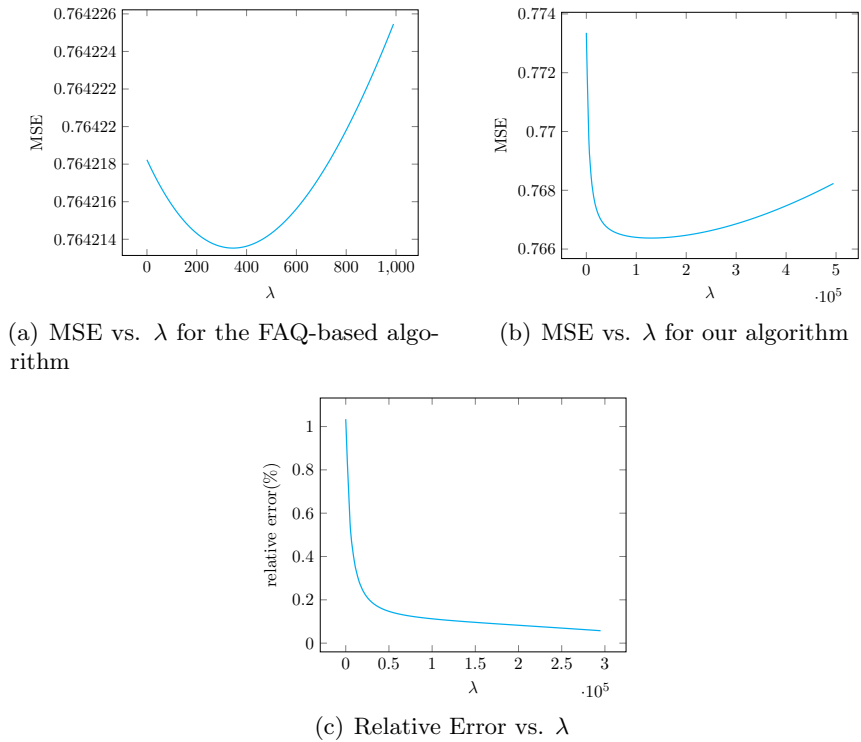
\begin{figure*}[ht!]
    \centering
    \subfigure[MSE vs. $\lambda$ for the FAQ-based algorithm]{
        \resizebox{!}{4cm}{
            \begin{tikzpicture}
            \begin{axis} [xlabel=$\lambda$, ylabel=MSE, yticklabel style={/pgf/number format/precision=6}]
            \addplot[smooth, thick, cyan] coordinates {
            (0, 0.7642182739031661)
            (10, 0.7642179905844355)
            (20, 0.7642177169904162)
            (30, 0.7642174530266036)
            (40, 0.7642171985994681)
            (50, 0.764216953616335)
            (60, 0.7642167179855021)
            (70, 0.7642164916161771)
            (80, 0.764216274418471)
            (90, 0.7642160663033626)
            (100, 0.7642158671827768)
            (110, 0.7642156769694591)
            (120, 0.7642154955770075)
            (130, 0.7642153229199405)
            (140, 0.7642151589135795)
            (150, 0.7642150034740828)
            (160, 0.7642148565184186)
            (170, 0.7642147179644285)
            (180, 0.7642145877307008)
            (190, 0.7642144657366784)
            (200, 0.7642143519025584)
            (210, 0.7642142461493336)
            (220, 0.7642141483987609)
            (230, 0.7642140585733829)
            (240, 0.7642139765964929)
            (250, 0.7642139023921092)
            (260, 0.764213835885035)
            (270, 0.7642137770007574)
            (280, 0.7642137256655269)
            (290, 0.7642136818063124)
            (300, 0.7642136453507603)
            (310, 0.7642136162272426)
            (320, 0.7642135943648014)
            (330, 0.7642135796932172)
            (340, 0.7642135721429043)
            (350, 0.764213571644931)
            (360, 0.7642135781310893)
            (370, 0.764213591533799)
            (380, 0.7642136117861398)
            (390, 0.764213638821819)
            (400, 0.7642136725751641)
            (410, 0.7642137129812143)
            (420, 0.7642137599755191)
            (430, 0.7642138134943374)
            (440, 0.7642138734744942)
            (450, 0.764213939853421)
            (460, 0.7642140125691687)
            (470, 0.764214091560318)
            (480, 0.7642141767661375)
            (490, 0.7642142681263836)
            (500, 0.7642143655814406)
            (510, 0.7642144690722317)
            (520, 0.7642145785402362)
            (530, 0.7642146939274915)
            (540, 0.7642148151766153)
            (550, 0.764214942230727)
            (560, 0.7642150750335135)
            (570, 0.7642152135291704)
            (580, 0.7642153576624393)
            (590, 0.7642155073785919)
            (600, 0.7642156626233627)
            (610, 0.7642158233430528)
            (620, 0.7642159894844293)
            (630, 0.7642161609948073)
            (640, 0.7642163378219189)
            (650, 0.7642165199140936)
            (660, 0.7642167072200531)
            (670, 0.7642168996890064)
            (680, 0.764217097270729)
            (690, 0.7642172999153325)
            (700, 0.764217507573464)
            (710, 0.7642177201962701)
            (720, 0.7642179377352649)
            (730, 0.7642181601424879)
            (740, 0.7642183873703572)
            (750, 0.7642186193717934)
            (760, 0.7642188561001345)
            (770, 0.7642190975091281)
            (780, 0.7642193435529723)
            (790, 0.7642195941862905)
            (800, 0.7642198493641101)
            (810, 0.7642201090418776)
            (820, 0.7642203731754852)
            (830, 0.7642206417211561)
            (840, 0.764220914635613)
            (850, 0.7642211918759193)
            (860, 0.764221473399514)
            (870, 0.7642217591642567)
            (880, 0.7642220491284054)
            (890, 0.7642223432506092)
            (900, 0.7642226414898372)
            (910, 0.7642229438054771)
            (920, 0.7642232501573122)
            (930, 0.7642235605054214)
            (940, 0.764223874810317)
            (950, 0.7642241930328406)
            (960, 0.764224515134204)
            (970, 0.7642248410759563)
            (980, 0.7642251708199791)
            (990, 0.7642255043285509)
            };
            \end{axis}
            \end{tikzpicture}
        }
        \label{fig:plot_brute_force}
    }
    \quad
    \subfigure[MSE vs. $\lambda$ for our algorithm]{
        \resizebox{!}{4cm}{
            \begin{tikzpicture}
            \begin{axis} [xlabel=$\lambda$, ylabel=MSE, yticklabel style={/pgf/number format/precision=3}]
            \addplot[smooth, thick, cyan] coordinates {
            (0, 0.7733565672426844)
            (5000, 0.7697418566571624)
            (10000, 0.7684056576692502)
            (15000, 0.7677521036397992)
            (20000, 0.7673765749881626)
            (25000, 0.767136846993459)
            (30000, 0.7669720500618807)
            (35000, 0.766852418218315)
            (40000, 0.7667619152845634)
            (45000, 0.7666912400746394)
            (50000, 0.7666346741141115)
            (55000, 0.766588531496673)
            (60000, 0.7665503415408141)
            (65000, 0.7665183925000426)
            (70000, 0.7664914643918648)
            (75000, 0.7664686661024908)
            (80000, 0.7664493325538055)
            (85000, 0.7664329577939037)
            (90000, 0.7664191502914537)
            (95000, 0.7664076023554779)
            (100000, 0.7663980687753785)
            (105000, 0.7663903516179067)
            (110000, 0.7663842892214451)
            (115000, 0.766379748105587)
            (120000, 0.7663766169393081)
            (125000, 0.7663748019861365)
            (130000, 0.7663742236238817)
            (135000, 0.766374813655998)
            (140000, 0.766376513213868)
            (145000, 0.7663792711043379)
            (150000, 0.7663830424969686)
            (155000, 0.7663877878724394)
            (160000, 0.7663934721737666)
            (165000, 0.7664000641162384)
            (170000, 0.7664075356229152)
            (175000, 0.7664158613593229)
            (180000, 0.7664250183482274)
            (185000, 0.766434985648567)
            (190000, 0.7664457440864232)
            (195000, 0.7664572760286171)
            (200000, 0.7664695651909709)
            (205000, 0.7664825964756319)
            (210000, 0.7664963558318745)
            (215000, 0.7665108301369634)
            (220000, 0.7665260070935339)
            (225000, 0.7665418751410102)
            (230000, 0.7665584233786127)
            (235000, 0.7665756414985042)
            (240000, 0.7665935197273651)
            (245000, 0.7666120487752366)
            (250000, 0.7666312197903942)
            (255000, 0.7666510243198639)
            (260000, 0.7666714542742332)
            (265000, 0.7666925018965892)
            (270000, 0.7667141597349304)
            (275000, 0.7667364206174424)
            (280000, 0.7667592776306229)
            (285000, 0.7667827240993516)
            (290000, 0.7668067535694375)
            (295000, 0.7668313597914156)
            (300000, 0.7668565367064044)
            (305000, 0.7668822784329564)
            (310000, 0.7669085792552941)
            (315000, 0.7669354336127223)
            (320000, 0.7669628360897915)
            (325000, 0.7669907814074731)
            (330000, 0.7670192644151076)
            (335000, 0.7670482800829139)
            (340000, 0.7670778234953045)
            (345000, 0.7671078898445564)
            (350000, 0.7671384744251625)
            (355000, 0.7671695726285089)
            (360000, 0.7672011799380241)
            (365000, 0.7672332919246594)
            (370000, 0.7672659042427465)
            (375000, 0.7672990126261509)
            (380000, 0.7673326128846104)
            (385000, 0.7673667009005145)
            (390000, 0.7674012726257281)
            (395000, 0.7674363240786394)
            (400000, 0.7674718513416685)
            (405000, 0.7675078505585319)
            (410000, 0.7675443179319832)
            (415000, 0.7675812497215633)
            (420000, 0.7676186422415465)
            (425000, 0.7676564918590081)
            (430000, 0.7676947949918997)
            (435000, 0.7677335481074486)
            (440000, 0.7677727477204608)
            (445000, 0.76781239039174)
            (450000, 0.7678524727267573)
            (455000, 0.7678929913741364)
            (460000, 0.767933943024485)
            (465000, 0.767975324409103)
            (470000, 0.7680171322987529)
            (475000, 0.7680593635027017)
            (480000, 0.7681020148675259)
            (485000, 0.7681450832761861)
            (490000, 0.7681885656470402)
            (495000, 0.7682324589329722)
            };
            \end{axis}
            \end{tikzpicture}
        }
        \label{fig:plot_sketch}
    }
    \quad
    \subfigure[Relative Error vs. $\lambda$]{
        \resizebox{!}{4cm}{
            \begin{tikzpicture}
            \begin{axis} [xlabel=$\lambda$, ylabel=relative error($\%$)]
            \addplot[smooth, thick, cyan] coordinates {
            (0, 1.0340838221929216)
            (5000, 0.5533044752699334)
            (10000, 0.3764104386335454)
            (15000, 0.2904845255714994)
            (20000, 0.241535426982864)
            (25000, 0.21058222531169513)
            (30000, 0.1895004989622251)
            (35000, 0.1743171920627562)
            (40000, 0.16289209607147637)
            (45000, 0.15398425488173384)
            (50000, 0.14683115740643515)
            (55000, 0.14094176609185105)
            (60000, 0.13598736669806488)
            (65000, 0.13174062298790812)
            (70000, 0.12803988540581734)
            (75000, 0.12476742665201712)
            (80000, 0.12183570039727698)
            (85000, 0.11917839965103205)
            (90000, 0.1167444819899277)
            (95000, 0.11449408265551408)
            (100000, 0.11239565999561218)
            (105000, 0.11042396394194132)
            (110000, 0.10855856563141142)
            (115000, 0.1067827767529561)
            (120000, 0.10508284411472335)
            (125000, 0.10344734172833547)
            (130000, 0.10186670645480067)
            (135000, 0.10033287956348058)
            (140000, 0.09883902719627358)
            (145000, 0.09737932038398345)
            (150000, 0.09594876041234546)
            (155000, 0.09454303906593076)
            (160000, 0.09315842598378588)
            (165000, 0.09179167716864001)
            (170000, 0.09043996022903844)
            (175000, 0.08910079288944051)
            (180000, 0.08777199214056175)
            (185000, 0.08645163192464445)
            (190000, 0.08513800779614833)
            (195000, 0.08382960716113441)
            (200000, 0.08252508422030846)
            (205000, 0.08122323866934789)
            (210000, 0.07992299760135602)
            (215000, 0.07862340001338897)
            (220000, 0.0773235835005659)
            (225000, 0.07602277280698289)
            (230000, 0.07472026993873282)
            (235000, 0.07341544555030666)
            (240000, 0.07210773149910477)
            (245000, 0.07079661430642314)
            (250000, 0.06948162944726555)
            (255000, 0.06816235633739787)
            (260000, 0.06683841393453349)
            (265000, 0.06550945682539666)
            (270000, 0.06417517178551879)
            (275000, 0.06283527472794592)
            (280000, 0.06148950798177122)
            (285000, 0.06013763785481796)
            (290000, 0.05877945250942762)
            (295000, 0.05741475998082279)
            };
            \end{axis}
            \end{tikzpicture}
        }
        \label{fig:plot_obj}
    }
    \caption{Experimental Results for General Joins}
\end{figure*}

%% file: Comparison.tex
\section{Connecting the Two Algorithms}

In this paper, we described two algorithms for computing subspace embeddings for database joins. The first applied only to the case of two-table joins, and the second applied to general joins queries. However, one may ask whether there are theoretical of empirical reasons to use the latter over the former, even for the case of two-table joins. In this section, we briefly compare the two algorithms in the context of two-table joins to address this question.

\subsection{Theoretical Comparison}

By Theorem \ref{thm:subspace}, the total running time to obtain a subspace embedding is the minimum of $\Tilde{O}((n_1 + n_2)d/\eps^2 + d^3 / \eps^2)$ and $\Tilde{O}((\nnz(\T_1) + \nnz(\T_2))/\eps^2 +(n_1 + n_2)/\eps^2 + d^5/\eps^2)$ using Algorithm \ref{alg:1}.

Now we consider the algorithm stated in Corollary \ref{corollary:general:join}. When running on the join of two tables, the algorithm is equivalent to applying the Fast Tensor-Sketch to each block and summing them up. Thus, the running time is  $\Tilde{O}(\sum_{i\in\mathcal{B}}(\nnz(\T_1^{(i)})/\eps^2 + \nnz(\T_2^{(i)})/\eps^2 + d^2/\eps^2)) = \Tilde{O}((\nnz(\T_1) + \nnz(\T_2))/\eps^2 + (n_1 + n_2)d^2/\eps^2)$.

The running time of the algorithm in Corollary \ref{corollary:general:join} is greater than the running time of Algorithm \ref{alg:1}. In the extreme case the number of blocks can be really large, and even if each block has only a few rows we still need to pay an extra $\Tilde{O}(d^2/\eps^2)$ time for it. This is the reason why we split the blocks into two sets ($\mathcal{B_{\text{big}}}$ and $\mathcal{B_{\text{small}}}$) and use a different approach when designing the algorithm for two-table joins.

\subsection{Experimental Comparison}

In our experiments we replace the Fast Tensor-Sketch with Tensor-Sketch. The theoretical analysis is similar since we still need to pay an extra $O(kd\log k)$ time to sketch a block for target dimension $k$.

We run the algorithm for the general case on joins $\J_1$ and $\J_2$ for different target dimensions. As shown in Table \ref{tab:Experiment:J_1,J_2-general}, in order to achieve the same relative error as Algorithm \ref{alg:1}, we need to set a large target dimension and the running time would be significantly greater than it is in Table \ref{tab:Experiment_two-table}, even compared with the  FAQ-based algorithm. This experimental result agrees with our theoretical analysis.

\begin{table}
\centering
\caption{General Algorithm on Two-Table Joins}
\label{tab:Experiment:J_1,J_2-general}
\begin{tabular}{|c|c|c|c|}
\hline
    & $k$ & running time & relative error\\
\hline
    \multirow{5}{*}{$\J_1$} & 40 & 0.086 & 42.3\% \\
\cline{2-4}
    & 80 & 0.10 & 9.79\% \\
\cline{2-4}
    & 120 & 0.14 & 3.67\% \\
\cline{2-4}
    & 160 & 0.16 & 0.87\% \\
\cline{2-4}
    & 200 & 0.19& 1.05\% \\
\hline
    \multirow{5}{*}{$\J_2$} & 400 & 1.25 & 5.79\% \\
\cline{2-4}
    & 800 & 2.02 & 2.32\% \\
\cline{2-4}
    & 1200 & 2.93 & 1.85\% \\
\cline{2-4}
    & 1600 & 3.74 & 1.19\% \\
\cline{2-4}
    & 2000 & 4.50 & 0.96\% \\
\hline
\end{tabular}
\end{table}

%% file: main.bbl
\newcommand{\etalchar}[1]{$^{#1}$}
\begin{thebibliography}{AKNN{\etalchar{+}}18}

\bibitem[ACJR19]{Arenas19}
Marcelo Arenas, Luis~Alberto Croquevielle, Rajesh Jayaram, and Cristian
  Riveros.
\newblock Efficient logspace classes for enumeration, counting, and uniform
  generation.
\newblock In {\em Proceedings of the 38th ACM SIGMOD-SIGACT-SIGAI Symposium on
  Principles of Database Systems}, PODS '19, pages 59--73, New York, NY, USA,
  2019. ACM.

\bibitem[ACJR21a]{10.1145/3406325.3465353}
Marcelo Arenas, Luis~Alberto Croquevielle, Rajesh Jayaram, and Cristian
  Riveros.
\newblock A polynomial-time approximation algorithm for counting words accepted
  by an nfa (invited paper).
\newblock In {\em Proceedings of the 53rd Annual ACM SIGACT Symposium on Theory
  of Computing}, STOC 2021, page~4, New York, NY, USA, 2021. Association for
  Computing Machinery.

\bibitem[ACJR21b]{Arenas2020when}
Marcelo Arenas, Luis~Alberto Croquevielle, Rajesh Jayaram, and Cristian
  Riveros.
\newblock When is approximate counting for conjunctive queries tractable?
\newblock {\em Proceedings of the Fifty-third Annual ACM Symposium on Theory of
  Computing (STOC))}, 2021.

\bibitem[AGMS02]{agms02}
Noga Alon, Phillip~B. Gibbons, Yossi Matias, and Mario Szegedy.
\newblock Tracking join and self-join sizes in limited storage.
\newblock {\em J. Comput. Syst. Sci.}, 64(3):719--747, 2002.

\bibitem[AK19]{ahle2019almost}
Thomas~D Ahle and Jakob~BT Knudsen.
\newblock Almost optimal tensor sketch.
\newblock {\em arXiv preprint arXiv:1909.01821}, 2019.

\bibitem[AKK{\etalchar{+}}20]{ahle2020oblivious}
Thomas~D Ahle, Michael Kapralov, Jakob~BT Knudsen, Rasmus Pagh, Ameya
  Velingker, David~P Woodruff, and Amir Zandieh.
\newblock Oblivious sketching of high-degree polynomial kernels.
\newblock In {\em Proceedings of the Fourteenth Annual ACM-SIAM Symposium on
  Discrete Algorithms}, pages 141--160. SIAM, 2020.

\bibitem[AKNN{\etalchar{+}}18]{Indatabase:Linear:Regression}
Mahmoud Abo~Khamis, Hung~Q. Ngo, XuanLong Nguyen, Dan Olteanu, and Maximilian
  Schleich.
\newblock In-database learning with sparse tensors.
\newblock In {\em Proceedings of the 37th ACM SIGMOD-SIGACT-SIGAI Symposium on
  Principles of Database Systems}, SIGMOD/PODS '18, pages 325--340, New York,
  NY, USA, 2018. ACM.

\bibitem[AKNR16]{abo2016faq}
Mahmoud Abo~Khamis, Hung~Q Ngo, and Atri Rudra.
\newblock Faq: questions asked frequently.
\newblock In {\em Proceedings of the 35th ACM SIGMOD-SIGACT-SIGAI Symposium on
  Principles of Database Systems}, pages 13--28. ACM, 2016.

\bibitem[AM00]{aji2000generalized}
Srinivas~M Aji and Robert~J McEliece.
\newblock The generalized distributive law.
\newblock {\em IEEE transactions on Information Theory}, 46(2):325--343, 2000.

\bibitem[AMS99]{ams99}
Noga Alon, Yossi Matias, and Mario Szegedy.
\newblock The space complexity of approximating the frequency moments.
\newblock {\em Journal of Computer and system sciences}, 58(1):137--147, 1999.

\bibitem[ANW14]{avron2014subspace}
Haim Avron, Huy Nguyen, and David Woodruff.
\newblock Subspace embeddings for the polynomial kernel.
\newblock In {\em Advances in neural information processing systems}, pages
  2258--2266, 2014.

\bibitem[ASW13]{asw13}
Haim Avron, Vikas Sindhwani, and David~P. Woodruff.
\newblock Sketching structured matrices for faster nonlinear regression.
\newblock In {\em Advances in Neural Information Processing Systems 26: 27th
  Annual Conference on Neural Information Processing Systems 2013. Proceedings
  of a meeting held December 5-8, 2013, Lake Tahoe, Nevada, United States},
  pages 2994--3002, 2013.

\bibitem[BCW19]{Bakshi2019RobustAS}
Ainesh Bakshi, Nadiia Chepurko, and David~P. Woodruff.
\newblock Robust and sample optimal algorithms for psd low-rank approximation.
\newblock {\em ArXiv}, abs/1912.04177, 2019.

\bibitem[Big]{BigqueryML}
\url{https://cloud.google.com/bigquery-ml/docs/bigqueryml-intro}.

\bibitem[BW18]{bakshi2018sublinear}
Ainesh Bakshi and David Woodruff.
\newblock Sublinear time low-rank approximation of distance matrices.
\newblock In {\em Advances in Neural Information Processing Systems}, pages
  3782--3792, 2018.

\bibitem[CBK11]{Cantador:RecSys2011}
Iv\'{a}n Cantador, Peter Brusilovsky, and Tsvi Kuflik.
\newblock 2nd workshop on information heterogeneity and fusion in recommender
  systems (hetrec 2011).
\newblock In {\em Proceedings of the 5th ACM conference on Recommender
  systems}, RecSys 2011, New York, NY, USA, 2011. ACM.

\bibitem[CK19]{cheng2019nonlinear}
Zhaoyue Cheng and Nick Koudas.
\newblock Nonlinear models over normalized data.
\newblock In {\em 2019 IEEE 35th International Conference on Data Engineering
  (ICDE)}, pages 1574--1577. IEEE, 2019.

\bibitem[CLM{\etalchar{+}}15]{cohen2015uniform}
Michael~B Cohen, Yin~Tat Lee, Cameron Musco, Christopher Musco, Richard Peng,
  and Aaron Sidford.
\newblock Uniform sampling for matrix approximation.
\newblock In {\em Proceedings of the 2015 Conference on Innovations in
  Theoretical Computer Science}, pages 181--190, 2015.

\bibitem[CW13]{10.1145/2488608.2488620}
Kenneth~L. Clarkson and David~P. Woodruff.
\newblock Low rank approximation and regression in input sparsity time.
\newblock In {\em Proceedings of the Forty-Fifth Annual ACM Symposium on Theory
  of Computing}, STOC ’13, page 81–90, New York, NY, USA, 2013. Association
  for Computing Machinery.

\bibitem[DDH07]{demmel2007fast}
James Demmel, Ioana Dumitriu, and Olga Holtz.
\newblock Fast linear algebra is stable.
\newblock {\em Numerische Mathematik}, 108(1):59--91, 2007.

\bibitem[Dec96]{Dechter:1996}
Rina Dechter.
\newblock Bucket elimination: A unifying framework for probabilistic inference.
\newblock In {\em Proceedings of the Twelfth International Conference on
  Uncertainty in Artificial Intelligence}, 1996.

\bibitem[DJS{\etalchar{+}}19]{diao2019optimal}
Huaian Diao, Rajesh Jayaram, Zhao Song, Wen Sun, and David~P. Woodruff.
\newblock Optimal sketching for kronecker product regression and low rank
  approximation, 2019.

\bibitem[DSSW17]{diao2018}
Huaian Diao, Zhao Song, Wen Sun, and David~P. Woodruff.
\newblock Sketching for kronecker product regression and p-splines.
\newblock {\em CoRR}, abs/1712.09473, 2017.

\bibitem[ELB{\etalchar{+}}17]{elgamal2017spoof}
Tarek Elgamal, Shangyu Luo, Matthias Boehm, Alexandre~V Evfimievski, Shirish
  Tatikonda, Berthold Reinwald, and Prithviraj Sen.
\newblock Spoof: Sum-product optimization and operator fusion for large-scale
  machine learning.
\newblock In {\em CIDR}, 2017.

\bibitem[FGR{\v{Z}}21]{focke2021approximately}
Jacob Focke, Leslie~Ann Goldberg, Marc Roth, and Stanislav {\v{Z}}ivn{\`y}.
\newblock Approximately counting answers to conjunctive queries with
  disequalities and negations.
\newblock {\em arXiv preprint arXiv:2103.12468}, 2021.

\bibitem[GM06]{GM06}
Martin Grohe and D{\'a}niel Marx.
\newblock Constraint solving via fractional edge covers.
\newblock In {\em SODA}, pages 289--298, 2006.

\bibitem[GWWZ15]{GWWZ15}
Dirk~Van Gucht, Ryan Williams, David~P. Woodruff, and Qin Zhang.
\newblock The communication complexity of distributed set-joins with
  applications to matrix multiplication.
\newblock In {\em Proceedings of the 34th {ACM} Symposium on Principles of
  Database Systems, {PODS} 2015, Melbourne, Victoria, Australia, May 31 - June
  4, 2015}, pages 199--212, 2015.

\bibitem[HAL{\etalchar{+}}20]{hu2019difftaichi}
Yuanming Hu, Luke Anderson, Tzu-Mao Li, Qi~Sun, Nathan Carr, Jonathan
  Ragan-Kelley, and Fr{\'e}do Durand.
\newblock Difftaichi: Differentiable programming for physical simulation.
\newblock {\em ICLR}, 2020.

\bibitem[HK15]{10.1145/2827872}
F.~Maxwell Harper and Joseph~A. Konstan.
\newblock The movielens datasets: History and context.
\newblock {\em ACM Trans. Interact. Intell. Syst.}, 5(4), December 2015.

\bibitem[HLA{\etalchar{+}}19]{hu2019taichi}
Yuanming Hu, Tzu-Mao Li, Luke Anderson, Jonathan Ragan-Kelley, and Fr{\'e}do
  Durand.
\newblock Taichi: a language for high-performance computation on spatially
  sparse data structures.
\newblock {\em ACM Transactions on Graphics (TOG)}, 38(6):201, 2019.

\bibitem[HRS{\etalchar{+}}12]{hellerstein2012madlib}
Joe Hellerstein, Christopher R{\'e}, Florian Schoppmann, Daisy~Zhe Wang, Eugene
  Fratkin, Aleksander Gorajek, Kee~Siong Ng, Caleb Welton, Xixuan Feng, Kun Li,
  et~al.
\newblock The madlib analytics library or mad skills, the sql.
\newblock {\em arXiv preprint arXiv:1208.4165}, 2012.

\bibitem[IVWW19]{indyk2019sample}
Piotr Indyk, Ali Vakilian, Tal Wagner, and David Woodruff.
\newblock Sample-optimal low-rank approximation of distance matrices.
\newblock {\em arXiv preprint arXiv:1906.00339}, 2019.

\bibitem[KJY{\etalchar{+}}15]{kumar2015demonstration}
Arun Kumar, Mona Jalal, Boqun Yan, Jeffrey Naughton, and Jignesh~M Patel.
\newblock Demonstration of santoku: optimizing machine learning over normalized
  data.
\newblock {\em Proceedings of the VLDB Endowment}, 8(12):1864--1867, 2015.

\bibitem[KNN{\etalchar{+}}18]{khamis2018ac}
Mahmoud~Abo Khamis, Hung~Q Ngo, XuanLong Nguyen, Dan Olteanu, and Maximilian
  Schleich.
\newblock Ac/dc: in-database learning thunderstruck.
\newblock In {\em Proceedings of the Second Workshop on Data Management for
  End-To-End Machine Learning}, page~8. ACM, 2018.

\bibitem[KNP15]{Kumar:2015:LGL:2723372.2723713}
Arun Kumar, Jeffrey Naughton, and Jignesh~M. Patel.
\newblock Learning generalized linear models over normalized data.
\newblock In {\em ACM SIGMOD International Conference on Management of Data},
  pages 1969--1984, 2015.

\bibitem[KNPZ16]{Kumar:2016:JJT:2882903.2882952}
Arun Kumar, Jeffrey Naughton, Jignesh~M. Patel, and Xiaojin Zhu.
\newblock To join or not to join?: Thinking twice about joins before feature
  selection.
\newblock In {\em International Conference on Management of Data}, pages
  19--34, 2016.

\bibitem[KW08]{Kohlas:2008}
J.~Kohlas and N.~Wilson.
\newblock Semiring induced valuation algebras: Exact and approximate local
  computation algorithms.
\newblock {\em Artif. Intell.}, 172(11):1360--1399, 2008.

\bibitem[LCK19]{li2019enabling}
Side Li, Lingjiao Chen, and Arun Kumar.
\newblock Enabling and optimizing non-linear feature interactions in factorized
  linear algebra.
\newblock In {\em Proceedings of the 2019 International Conference on
  Management of Data}, pages 1571--1588, 2019.

\bibitem[LMP13]{li2013iterative}
Mu~Li, Gary~L Miller, and Richard Peng.
\newblock Iterative row sampling.
\newblock In {\em 2013 IEEE 54th Annual Symposium on Foundations of Computer
  Science}, pages 127--136. IEEE, 2013.

\bibitem[MW17]{musco2017sublinear}
Cameron Musco and David~P Woodruff.
\newblock Sublinear time low-rank approximation of positive semidefinite
  matrices.
\newblock In {\em 2017 IEEE 58th Annual Symposium on Foundations of Computer
  Science (FOCS)}, pages 672--683. IEEE, 2017.

\bibitem[NN13]{nelson2013osnap}
Jelani Nelson and Huy~L Nguy{\^e}n.
\newblock Osnap: Faster numerical linear algebra algorithms via sparser
  subspace embeddings.
\newblock In {\em 2013 ieee 54th annual symposium on foundations of computer
  science}, pages 117--126. IEEE, 2013.

\bibitem[Pag13]{pagh2013compressed}
Rasmus Pagh.
\newblock Compressed matrix multiplication.
\newblock {\em ACM Transactions on Computation Theory (TOCT)}, 5(3):1--17,
  2013.

\bibitem[PP13]{pham2013fast}
Ninh Pham and Rasmus Pagh.
\newblock Fast and scalable polynomial kernels via explicit feature maps.
\newblock In {\em Proceedings of the 19th ACM SIGKDD international conference
  on Knowledge discovery and data mining}, pages 239--247, 2013.

\bibitem[Rel]{RelationalAI}
\url{https://www.relational.ai/}.

\bibitem[Ren13]{rendle2013scaling}
Steffen Rendle.
\newblock Scaling factorization machines to relational data.
\newblock In {\em Proceedings of the VLDB Endowment}, volume~6, pages 337--348.
  VLDB Endowment, 2013.

\bibitem[SOC16]{SystemF}
Maximilian Schleich, Dan Olteanu, and Radu Ciucanu.
\newblock Learning linear regression models over factorized joins.
\newblock In {\em Proceedings of the 2016 International Conference on
  Management of Data}, SIGMOD '16, pages 3--18. ACM, 2016.

\bibitem[SW19]{ShiW19}
Xiaofei Shi and David~P. Woodruff.
\newblock Sublinear time numerical linear algebra for structured matrices.
\newblock In {\em The Thirty-Third {AAAI} Conference on Artificial
  Intelligence}, pages 4918--4925, 2019.

\bibitem[VL00]{van2000ubiquitous}
Charles~F Van~Loan.
\newblock The ubiquitous kronecker product.
\newblock {\em Journal of computational and applied mathematics},
  123(1-2):85--100, 2000.

\bibitem[W{\etalchar{+}}14]{woodruff2014sketching}
David~P Woodruff et~al.
\newblock Sketching as a tool for numerical linear algebra.
\newblock {\em Foundations and Trends{\textregistered} in Theoretical Computer
  Science}, 10(1--2):1--157, 2014.

\bibitem[WZ20]{Woodruff2020near}
David~P. Woodruff and Amir Zandieh.
\newblock Near input sparsity time kernel embeddings via adaptive sampling.
\newblock In {\em International Conference on Machine Learning (ICML)}, 2020.

\bibitem[YGL{\etalchar{+}}]{yangtowards}
Keyu Yang, Yunjun Gao, Lei Liang, Bin Yao, Shiting Wen, and Gang Chen.
\newblock Towards factorized svm with gaussian kernels over normalized data.

\end{thebibliography}
